\documentclass[a4paper,11pt]{article}

\usepackage{fullpage}
\usepackage[colorlinks,urlcolor=blue,citecolor=blue,linkcolor=blue]{hyperref}

\usepackage{amsmath}
\usepackage{color}
\usepackage{amsthm}
\usepackage{amsfonts}
\usepackage{subcaption}

\usepackage{thmtools} 
\usepackage{thm-restate}
\usepackage{lipsum}

\newcommand{\alg}[1]{
	\begin{figure*}[htb!]
		\fbox{
			\begin{minipage}{0.95\textwidth}{
					#1	
				}
			\end{minipage}
		}
	\end{figure*}
}

\newtheorem{thm}{Theorem}[section]
\newtheorem{corollary}[thm]{Corollary}
\newtheorem{definition}[thm]{Definition}
\newtheorem{lemma}[thm]{Lemma}

\newtheorem{notation}[thm]{Notation}
\newtheorem{theorem}[thm]{Theorem}

\newcommand{\mA}{\mathcal{A}}

\newcommand{\oh}{\cn H}
\newcommand{\ohest}{\lowercase{\tilde{\mathtt{h}}}}
\newcommand{\eps}{\epsilon}
\renewcommand{\th}{^{\textrm{th}}}

\newcommand{\mD}{{D^*}}

\newcommand{\mC}{\mathcal{C}}

\newcommand{\mG}{\mathcal{G}}
\newcommand{\dmax}{d_{max}}
\newcommand{\du}{d_{ub}}

\newcommand{\cnt}{\text{max-count}}

\newcommand{\Sp}{\cn s_p}
\newcommand{\hSp}{\chat s_p}
\newcommand{\dc}{\textproc{decomp-cost}}

\def\loglog{{\rm \;log\,log\;}}

\newcommand{\cn}[1]{{\lowercase{\bar{\mathtt{#1}}}}}
\newcommand{\chat}[1]{{\lowercase{\hat{\mathtt{#1}}}}}

\newcommand{\davg}{d_{\mathsf{avg}}}

\def\HH{\cn H}

\newcommand\numberthis{\addtocounter{equation}{1}\tag{\theequation}}

\newcommand{\SH}{\hyperref[ch]{\textproc{\color{black}{\sf Sample-$H$}}}}

\newcommand{\SoC}{\hyperref[soc]{\textproc{\color{black}{\sf Sample-Odd-Cycle}}}}

\newcommand{\lps}{\hyperref[lps]{\textproc{\color{black}{\sf Sublinear-$\ell_p$-Sampler}}}}
\newcommand{\SaS}{\hyperref[sac]{\textproc{\color{black}{\sf Sample-a-Star}}}}

\newcommand{\NO}{$\textsf{NO}$ }
\newcommand{\YES}{$\textsf{YES}$ }
\newcommand{\MM}{$\texttt{cycle-gadget}$}

\newcommand{\CC}{$\texttt{CC-gadget}$}
\newcommand{\SSS}{$\texttt{star-gadget}$}
\newcommand{\SCG}{$\texttt{few-cycles-gadget}$}

\newcommand{\rmax}{r_{\max}}

\newcommand\mydots{\hbox to 1em{.\hss.\hss.\hss}}
\newcommand{\xs}{x^{*}}

\newcommand{\decomp}{\{O_{k_1}, \mydots, O_{k_q}, S_{p_1}, \mydots, S_{p_\ell}\}}

\newcommand{\comp}{ \max_{i\in[r]}\left\{cost(C_i) \right\}  \cdot \frac{\prod \cn c_{i} }{\HH} }
\newcommand{\SD}{\textproc{Set-Disjointness}}

\newcommand{\tSD}{t\textproc{-Set-Disjointness}}

\newcommand\compldots{\hbox to 1em{.\hss.\hss.}}

\newcommand{\shankhaInlineTodo}[1]{}

\DeclareMathOperator{\poly}{poly}

\definecolor{mygreen}{RGB}{20,140,80}
\definecolor{linkcolor}{RGB}{0,0,230}
\definecolor{mylightgray}{RGB}{230,230,230}
\definecolor{verylightgray}{RGB}{245,245,245}

\usepackage{paralist}
\title{Towards a Decomposition-Optimal Algorithm for Counting and Sampling Arbitrary Motifs in Sublinear Time\footnote{A proceedings version of this paper is to appear in RANDOM 2021.}} %TODO Please add

\author{%
	 Amartya Shankha Biswas \thanks{CSAIL at MIT, \textit{asbiswas@mit.edu}. Big George Ventures Fund, MIT-IBM Watson AI Lab and Research Collaboration Agreement No. W1771646, NSF awards CCF-173380, CCF-2006664 and IIS-1741137.}
	\and Talya Eden \thanks{CSAIL at MIT, \textit{talyaa01@gmail.com}. This work was supported by the National Science Foundation under Grant No.  CCF-1740751, the Eric and Wendy Schmidt Fund for Strategic Innovation, and Ben-Gurion University of the Negev.}
	\and Ronitt Rubinfeld \thanks{CSAIL at MIT, \textit{ronitt@csail.mit.edu}. This work was supported 
		the NSF TRIPODS program (awards CCF-1740751 and DMS 2022448),
		NSF award CCF-2006664 and  by the Fintech@CSAIL initiative.}
}

\usepackage[dvipsnames]{xcolor}
\usepackage{paralist}
\usepackage{bbm}
\usepackage{nicefrac}
\usepackage[noend]{algpseudocode}

\usepackage{todonotes}
\usepackage{mathtools}
\usepackage[T1]{fontenc}
\usepackage[most]{tcolorbox}

\def\comments{1}

\usetikzlibrary{
	shapes.multipart,
	matrix,
	positioning,
	shapes.callouts,
	shapes.arrows,
	calc}

\allowdisplaybreaks

\if\comments
\newcommand{\shankhaOld}[1]{\textcolor{black!15}{#1}}
\newcommand{\talyaNote}[1]{\textcolor{blue!80!green!25}{#1}}
\newcommand{\talyaOld}[1]{\textcolor{black!15!}{#1}}

\newcommand{\Ttodo}[1]{\todo[color=blue!20,inline]{Talya: #1}}
\else
\newcommand{\shankhaOld}[1]{}
\newcommand{\talyaNote}[1]{}
\newcommand{\talyaOld}[1]{}
\newcommand{\Ttodo}[1]{}

\fi

\newtheorem*{theorem*}{Theorem}

\begin{document}
	
	%\color{blue}
	%\global\let\default@color\current@color
	\maketitle
	
	\begin{abstract}
We consider the problem of sampling and approximately counting an arbitrary given motif $H$ in a graph $G$, where access to $G$ is given via queries: degree, neighbor, and pair, as well as uniform edge sample queries.
Previous algorithms for these tasks were based on a decomposition of $H$ into a collection of odd cycles and stars, denoted $\mD(H)=\decomp$. These algorithms were shown to be optimal for the  case where $H$ is a clique or an odd-length  cycle, but no other  lower bounds were known.

We present a new  algorithm  for sampling and approximately counting arbitrary motifs which, up to $\poly(\log n)$ factors,  is always at least as good as previous results, and for most graphs $G$ is strictly better. 
The main ingredient leading to this improvement is an improved  uniform  algorithm for sampling stars, which might be of independent interest, as it allows to sample vertices according to the $p$-th moment of the degree distribution.

Finally, we prove that this algorithm is \emph{decomposition-optimal} for decompositions that contain at least one odd cycle. 
These are the first lower bounds for motifs $H$ with a nontrivial decomposition, i.e., motifs that have more than a single component in their decomposition.
	\end{abstract}
	
	\newcommand\fullversion{1}

	\section{Introduction}
The problems of counting and sampling small motifs in graphs are  fundamental algorithmic problems with many applications.
Small motifs statistics are used for the study and characterization of  graphs  in
multiple fields, including biology, chemistry, social networks and many others (see e.g., ~\cite{shen2002network,lee2002transcriptional,eichenberger2004program,odom2004control, nelson2004oscillations,zhao2010communication, juszczyszyn2008local,paranjape2017motifs, topirceanu2016uncovering,tyson2010functional,ma2009defining}). 
From a theoretical perspective, the complexity of the best known classical algorithms for {exactly} enumerating small motifs such as cliques and paths of length $k$, grows exponentially with $k$~\cite{V09, Bjorklund09}.
On the more applied side, there is an extensive study of practical algorithms for approximate motif counting (e.g.,~\cite{tsourakakis2008fast,avron2010counting,PaTs12,ahmed2015efficient, JS17,danisch2018listing,BeraPS20,FoxRSWW20}). 
We study the problems of approximate motif counting and uniform sampling in the \emph{sublinear-time} setting, where sublinear is with respect to the size of the graph. We consider the \emph{augmented query model}, introduced by \cite{counting_stars_edge_sampling}, where the allowed queries are degree, neighbor and pair queries as well as uniform edge sample queries.\footnote{Degree queries return the degree of the queried vertex,  neighbor queries with index $i\leq d(v)$ return the $i\th$ neighbor of the queried vertex, pair queries return whether there is an edge between the queried pair of vertices, and uniform edge queries return a uniformly distributed edge in the graph.} We note that the model which  only allows for the first three types of queries is referred to as the \emph{general graph query model}, introduced by~\cite{KKR04}.

The problems of approximate counting and uniformly sampling of \emph{arbitrary motifs} of constant size in sublinear-time have seen much progress recently,
through the results of Assadi, Kapralov and Khanna~\cite{AKK19}, and Fichtenberger, Gao and Peng~\cite{Peng}. The algorithms of \cite{AKK19, Peng} both start by computing an optimal (in a sense that will be clear shortly) decomposition  of the motif $H$ into vertex-disjoint odd cycles and stars, defined next.
% \iffalse{
% \textcolor{black!15}{We continue this line of work and present a new algorithm which, ignoring lower order terms, strictly improves on previous works.  We also provide new lower bounds for a much larger class of motifs $H$ and graphs $G$. Previously, the only known lower bounds were for motifs that are either an odd cycle or a clique.
% }
% \fi

\textbf{A decomposition into odd cycles and stars.} 
A decomposition $D$ of a motif (graph) $H$ into a collection  of vertex disjoint small cycles and stars $\decomp$ is valid if all vertices of $H$ belong to either a star or an odd cycle in the collection. 
Each decomposition can be associated with a weight function $f_D:E\rightarrow \{0,\frac{1}{2},1\}$ which assigns weight $1$ to edges of its star components, weight $1/2$ to edges of its odd cycle components and weight 0 to all other edges in $H$. See figure 1 for an illustration.
Hence, each decomposition $\decomp$ has value 
$\rho(D)=\sum_{e\in H}{f_D(e)}=\sum_{i=1}^q k_i/2+\sum_{j=1}^{\ell} p_j$,
where throughout the paper $k_i$ and $p_j$ denote the length and number of petals in the $i^{th}$ cycle and $j\th$ star, respectively, in $\mD(H)$. 
%An optimal decomposition for $H$ is one with value that equals the fractional edge cover value of $H$. 
For every $H$, its optimal decomposition value is $\rho(H)=\min_{D}\{\rho(D)\}$, and a decomposition $D$ is said to be \emph{optimal} for $H$ if $\rho(D)=\rho(H)$. We fix (one of) the optimal decomposition of $H$, and denote it by $\mD(H)$.
In~\cite{AKK19}, it is shown 
that an optimal decomposition of a motif $H$ can be computed in polynomial time in $|H|$.\footnote{We note that $\rho(H)$ is equal to the fractional edge cover value of $H$: the fractional edge cover value of a motif (graph) $H$ is the solution to the following minimization problem. Minimize $\sum_{e\in E}f(e)$ under the constraint that for every $v\in H$, $\sum_{e\ni v}f(e)\geq 1$. In~\cite{AKK19}, the decomposition is  computed by first computing an optimal fractional cover. However, as there exists a mapping between fractional edge covers to decompositions  which preservers their value, we choose to define $\rho(H)$ according to the minimal valid decomposition value.}

% has a decomposition $\decomp$ to vertex-disjoint odd cycles and stars such that $\sum_{f_D(e)}=\sum_{i=1}^q k_i/2+\sum_{i=1}^{\ell} p_i=\rho(H)$. We fix one such optimal decomposition, and  denote it by $\mD(H)$.  See Figure~1 for an illustration of a motif $H$ and its decomposition $\mD(H)$,
% Each decomposition is associated with a value  $\rho(\mD(H))=\sum_{f(e)}=\sum_{i=1}^q k_i/2+\sum_{i=1}^{\ell} p_i$,
% where throughout the paper $k_i$ denotes the length of the $i^{th}$ cycle, and $p_i$ denotes the number of petals in the $i^{th}$ star in $\mD(H)$. Let $\rho(H)=\min_{\mD(H)}\{\rho(\mD(H)))$.
% A decomposition is \emph{optimal} if its value is  $\rho(H)$, the minimal value of a fractional edge cover of $H$. 
%and see Section~\ref{sec:prel} for more details about the properties of the decomposition.

\begin{figure}[t]
	\centering
	\begin{subfigure}{.7\textwidth}	\label{fig:decomp}
		\centering
		\includegraphics[width=.9\textwidth]{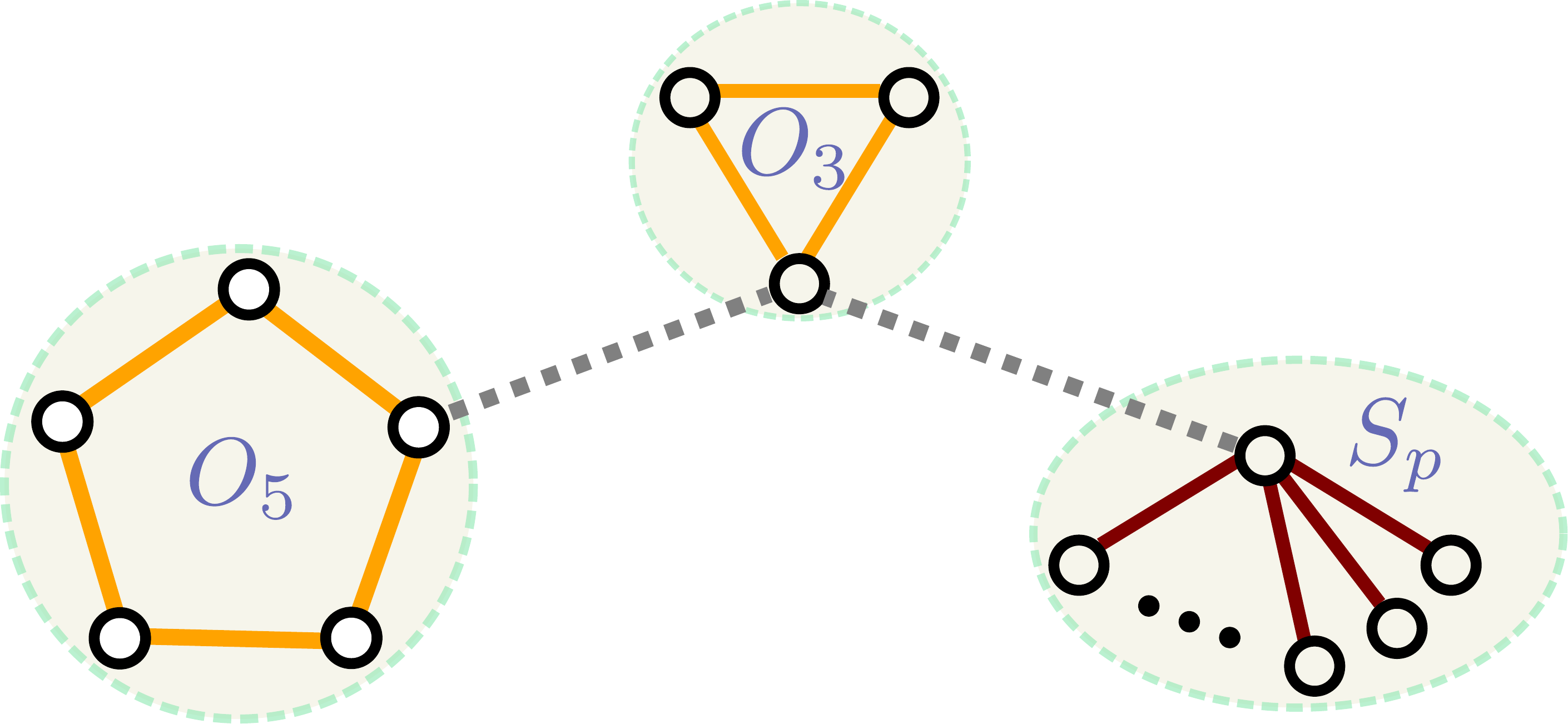}
	\end{subfigure}%
\caption{An example of an optimal decomposition of a motif $H$ into odd cycles and stars. The orange edges have weight $1/2$, the red edges have weight $1$, and the dotted edges have zero weight. }
\end{figure}

The algorithm in~\cite{Peng} has expected running time
\footnote{Throughout the paper, unless stated otherwise, the query complexity of  the mentioned  sublinear-time algorithms is the same as the minimum between their running time and $\min\{n+m, m\log n\}$. This is true since any algorithm  can simply query the entire graph and continue computation locally. 
	Querying the entire graph can either be performed by querying the neighbors of all vertices (which takes $O(n+m)$ queries), or by performing $m\log n$ uniform edge samples, which, with high probability, return all edges in the graph (note that we do not care about isolated vertices, as we assume the motif $H$ is connected). Hence, we focus our attention on the running time complexity.} 
 $O\left(\frac{m^{\rho(H)}}{\HH}\right)$ for the task of uniformly sampling a copy of $H$,
where $\HH$ is the number of copies of $H$ in $G$, and $m$ is the number of oriented edges\footnote{Throughout the paper we think of every edge $\{u,v\}$ as two oriented edges $(u,v)$ and $(v,u)$, and let $m$ denote the number of oriented edges.} in $G$. The algorithm in~\cite{AKK19} for the estimation task has the same complexity up to $\poly (\eps,|H|, \log n$) factors.

	\subsection{Our results}
We present improved upper and lower bounds for the tasks of estimating and sampling any arbitrary motif in a graph $G$ in sublinear time
(with respect to the size of $G$).
First, we give a new, essentially optimal, star-sampler for graphs. We also show that with few modifications, the star-sampler can be adapted to an optimal $\ell_p$ sampler, which  might be of independent interest.
Based on this sampler, as well as an improved sampling approach,
we present our main algorithm for sampling a uniformly distributed copy of any given motif $H$ in a graph $G$.
Our algorithm's complexity is parameterized by what we refer to as the \emph{decomposition-cost} of $H$ in $G$, denoted  $\dc(G,H, \mD(H))$. We further show that our motif sampling algorithm can be used to obtain a $(1\pm\eps)$-estimate of the motif at question (with an overhead of an  $O(1/\eps^2)$ factor).
 As we shall see, our result  is always at least as good as previous algorithms for these problems (up to a  $\log n\loglog n$ term), and greatly improves upon them for various interesting graph classes, such as random graphs and bounded arboricity graphs. 

We then continue to prove that for any motif whose optimal decomposition contains at least one odd cycle, this bound is  \emph{decomposition-optimal}:
we show that for every decomposition $D$ that contains at least one odd cycle, there exists a motif $H_{D}$ (with optimal decomposition $D$) and a family of graphs $\mG$ so that in order to sample a uniformly distributed copy of $H$ (or to approximate $\cn h$) in a uniformly chosen graph in $\mG$, the number of required queries is $\Omega(\min\{\dc(G, H, \mD(H)), m\})$ in expectation.

We start by describing the upper bound.

\subsubsection{Optimal star/$\ell_p$-sampler}
Our first contribution is an improved algorithm, \SaS,  for sampling a (single) star uniformly at random, and its variant for sampling vertices according to the $p\th$ moment.
For a vertex $v$, we let $\cn s_p(v)=\binom{d(v)}{p}$, if 
$d(v)\geq p$, and otherwise, $\cn s_p(v)=0$.
We let $\cn s_p=\sum_{v\in V} \cn s_p(v)$ denote the number of $p$-stars in the graph.
We will also be interested in the closely related value of the $p\th$ moment of the degree distribution,  $\cn \mu_p=\sum_{v \in V} d(v)^p$.

\begin{theorem}
	\label{lem:star-sampler}
	\sloppy
There exists a procedure, \SaS, that given query access to a graph $G$, and a constant factor estimates of $\cn s_p$, 
returns a uniformly distributed $p$-star in $G$.
%With probability at least\footnote{For all of our algorithms, the success probability can be increased to $1-\delta$ for any $\delta\in(0,1)$ while incurring a multiplicative cost of $\log(1/\delta)$ to the query complexity.} $2/3$ the returned star is uniformly distributed in $G$. 
The expected query complexity and running time of the procedure are  
$O\left(\min\left\{\frac{m\cdot n^{p-1}}{\Sp},\frac{m}{\Sp^{1/p}} \right\}\right)$ 
%\cdot {\log n\loglog n} \right),$
where $\cn s_p$ denotes the number of $p$-stars in $G$.
\end{theorem}

We note that  a constant factor estimate of  $\cn s_p$ can be obtained by invoking one of the algorithms in~\cite{ERS19-sidma, counting_stars_edge_sampling}, in expected query complexity  $\tilde{O}\left(\min\left\{\frac{m\cdot n^{p-1}}{\Sp},\frac{m}{\Sp^{1/p}} \right\}\right)$. Therefore, if such an estimate is not known in advance, then it could be computed, with probability at least $2/3$,   by only incurring a $\log n$  factor to the expected time complexity.

We will also show a variant of \SaS, denoted \lps, that  gives an optimal $\ell_p$-sampler for any integer $p\geq 2$ in sublinear time.
That is, \lps\ allows to sample according to the $p\th$ moment of the degree distribution, so that every vertex $v\in V$ is returned by it with probability $d(v)^p/\cn \mu_p$. 
The question of sampling according to the $p\th$ moment for various values of $p$ has been studied extensively in the streaming model where $\ell_p$ samplers have found numerous applications, see, e.g., the recent survey by Cormode and Hossein~\cite{cormode2019p} and the references therein. Therefore we hope it could find applications in the sublinear-time setting that go beyond subgraph sampling.

\begin{theorem}\label{thm:lpsamp}
There exists  an algorithm, \lps, that returns a vertex $v\in V$, so that each $v\in V$ is returned with probability $d(v)^p/\cn \mu_p$. The expected running time of the algorithm is  	$O\left(\min\left\{\frac{m\cdot n^{p-1}}{\cn \mu_p},\frac{m}{\cn \mu_p^{1/p}} \right\}\right)$ .
\end{theorem}

Observe that for every value of $p$, $\cn s_p < \cn \mu _p.$
 Furthermore, observe that  $m$ and $\cn \mu_p^{1/p}$  are simply the $\ell_1$ and $\ell_p$ norms of the degree distribution of $G$. Therefore,  it holds that $\cn \mu_p^{1/p}$ is  smaller than $m$ (and could be as small as $m/n^{1-1/p}$). 
% The strict inequality  is since two norms are equivalent only if they equal zero or if all the mass resides on a single element, which  never holds for a non-empty graph $G$. 
 Therefore,  $ \cn \mu_p^{1/p}<m \Leftrightarrow \mu_p^{p-1/p}< m^{p-1} $.
 and it follows that 
 \begin{align} \label{eqn:sp_m}
 m\cdot \min\left\{n^{p-1},\cn s^{(p-1)/p}\right\} \leq
  m\cdot \cn s_p^{(p-1)/p} <
  m\cdot \cn \mu_p^{(p-1)/p} \leq 
   m \cdot m^{p-1} =
    m^p. 
 \end{align}
\sloppy
 Hence, not accounting for the $O(\log n \loglog n)$ term, 
 the expected complexity  $\tilde{O}(m\cdot \min\{n^{p-1}, \cn s_p^{(p-1)/p}\}/\cn s_p)$
of  \SaS\ strictly improves upon the    $O(m^p/\cn s_p)$ expected complexity of the star-sampling algorithm by~\cite{Peng}.
 Accounting for that term, our algorithm is preferable when either $\davg=\omega(\log n\loglog n)$ or $m/\cn s_p^{1/p}=\omega( \log n)$.
 
Furthermore, the complexity of \SaS\  matches the complexities of the star approximation algorithms by~\cite{GRS11, counting_stars_edge_sampling}, thus proving that uniformly sampling and approximately counting stars in the augmented model have essentially the same complexity.
Finally, the  construction of the lower bound for the
estimation variant by~\cite{GRS11} proves that \SaS\  and \lps\ are essentially optimal.
%~\footnote{The lower bound proof in~\cite{GRS11} proves that an algorithm cannot distinguish with high probability between two families of graphs, one with  the minimum possible number of $p$-stars, and one having $\cn s_p$ stars, unless $O\left(\min\left\{\frac{m\cdot n^{p-1}}{\Sp},\frac{m}{\Sp^{1/p}} \right\}\right)$ queries are performed. 
%	Since any star- or $\ell_p$-sampling   algorithms can be used to distinguish such families in a constant number of invocations, their lower bound extends to the sampling tasks.
%} 

\subsubsection{An algorithm  for sampling and estimating  arbitrary motifs}

Given the above star sampler, we continue to describe our main contribution: an algorithm, \SH, that for any graph $G$ and given motif $H$,
outputs a uniformly distributed copy of $H$ in $G$.

To sample a copy of $H$ we first sample copies of all basic components in its decomposition $\mD(H)$, and then check if they can be extended to a copy of $H$ in $G$. Therefore, it will be useful to define the costs of these sampling operations.

\begin{notation}[Basic components, counts and costs]\label{def:costs}
Let $H$ be a motif, and let $\mD(H)=\decomp$ be an optimal decomposition of $H$.
We refer to the odd cycles and stars in $\mD(H)$ as the \emph{basic components} of the decomposition (or sometimes, abusing notation, of $H$). We use the notation $\{C_i\}_{i\in [r]}$, to denote the set of all components in $\mD(H)$, $\{C_i\}_{i\in [r]}=\mD(H)$, where $r=q+\ell$.

For every basic component $C_i$ in $\mD(H)=\{C_i\}_{i\in[r]}$, we denote  the number of copies of $C_i$ in $G$ as $\cn C_i$ and refer to it as the \emph{count} of $C_i$. Similarly, $\cn o_{k}$ and $\cn s_p$ denote the number of copies of length $k$ odd cycles and $p$-stars in $G$. respectively. 

We also define the \emph{sampling cost} (or just \emph{cost} in short) of $C_i$ to be:
\[
cost(C_i) = 
\begin{cases}
m^{k/2}/\cn o_{k} &  C_i=O_{k}\\
\min\left\{ \frac{m\cdot n^{p-1}}{\Sp}, \frac{m}{\Sp^{1/p}} \right\} & C_i=S_p
%\cdot \frac{\log n\loglog n}{\eps^2}
\end{cases}.
\]
\end{notation}

Observe that indeed, by Theorem~\ref{lem:star-sampler}, sampling a single $p$-star in $G$ takes $cost(S_p)=\min\left\{ \frac{m\cdot n^{p-1}}{\Sp}, \frac{m}{\Sp^{1/p}} \right\}$ queries in expectation, and by~\cite[Lemma 3.1]{Peng}, sampling a single $O_k$  odd cycle takes $cost(O_k)=m^{k/2}/\cn o_k$ queries in expectation.

\begin{notation}[Decomposition-cost]\label{def:decomp-cost}
	For a motif $H$, an optimal decomposition $\mD(H)$ of $H$, and a graph $G$, the \emph{decomposition cost} of $H$ in $G$, denoted $\dc(G,H, \mD(H))$ is 
%	the maximum of all single component costs, times the product of all counts of the basic components, divided by the number of copies of $H$ in $G$. That is, 
	\[
	\dc(G,H,\mD(H))=\comp \;.
	\]
	Note that the motif $H$ determines the counts of $\cn h$ and its decomposition $\mD(H)$ determines what are the basic component counts in $G$ that are relevant to the sampling cost.
\end{notation}

\begin{restatable}{theorem}{ubMain}
	\label{thm:intro_ub}
	Let $G$ be a graph over $n$ vertices and $m$ edges, and let  $H$ be a motif such that $\mD(H)=\decomp=\{C_i\}_{i\in[r]}$. 
		There exists an algorithm, \SH, that returns a  copy of $H$ in $G$.
	With probability at least $1-1/\poly(n)$, the returned copy is uniformly distributed in $G$.
	The expected query complexity of the algorithm is
	$$O\left( \min\left\{\dc(G,H,\mD(H)),m\right\}\right)\cdot {\log n\loglog n}.$$
\end{restatable}

%\textcolor{blue}{We note that since the error probability  of our algorithm is at most $1/\poly(n)$,   for any algorithm $\mA$ that invokes the sampler $\SH$ instead of a truly uniform sampler for $o(n^2)$, the success probability only changes by $o(1)$. For full details, see the discussion \textit{``On the (un)importance of sampling exactly"} in~\cite[Section 7, Theorem 8]{tetek}.}

\iffalse
\begin{theorem}
\label{thm:intro_ub}
Let $G$ be a graph over $n$ vertices and $m$ edges, and let  $H$ be a motif such that $\mD(H)=\decomp=\{C_i\}_{i\in[r]}$. 

There exists an algorithm, \SH, that returns a  copy of $H$.
With probability at least $2/3$, the returned copy is uniformly distributed in $G$.
The expected query complexity of the algorithm is
	$$O\left( \min\left\{\comp,m\right\}\cdot {\log n\loglog n}\right).$$
\end{theorem}
\fi

In 
\ifnum\fullversion=1
\ref{sec:from_sampling_to_estimating}, \ref{thm:samp_to_est}, 
\else
the full version
\fi
we prove that with slight modifications to the sampling algorithm we  can obtain a $(1\pm\eps)$-approximation algorithm for $\HH$, 
with the same  expected query complexity and running time
up to a multiplicative factor of $O(1/\eps^2)$.

\textbf{Comparison to previous bounds.}
We would like to compare our algorithm's expected complexity stated in Theorem~\ref{thm:intro_ub}, to the expected complexity $O\left(\frac{m^{\rho(H)}}{\cn h}\right)$ of the counting and sampling algorithms by~\cite{AKK19} and~\cite{Peng}, respectively,  where recall that for an optimal decomposition $\mD(H)=\decomp$ of $H$, $\rho(H)=\sum_{i\in [q[} k_i/2+\sum_{i\in[\ell]} p_i$.

Recalling Equation~\ref{eqn:sp_m}, and plugging in the costs of the basic components and the decomposition cost, defined in
Notations~\ref{def:costs} and~\ref{def:decomp-cost}, respectively, we get that 	for any graph $G$ and motif $H$,
\begin{align*} 
\dc(G,H,\mD(H))&=\comp 
\\& = \max_{i\in[r]}\{cost(C_i)\}\cdot \frac{\prod_{i\in[q]} \cn o_{k_i}\cdot \prod_{i\in[\ell]}  \cn s_{p_i}}{\cn h} 
\\ & \leq 
 \frac{\prod_{i\in[q]} m^{k_i/2}\cdot \prod_{i\in[\ell]} m\cdot  (\min\{n^{p_i-1},\cn s_{p_i}^{(p_i-1)/p_i}\})}{\HH}  \\
 &  < \frac{\prod_{i\in[q]} m^{k_i/2}\cdot \prod_{i\in[\ell]} m^p }{\HH}=\frac{m^{\rho(H)}}{\HH}, 
\end{align*}
Therefore, as long as $\mD(H)$ contains at least one star, and not accounting for the $O(\log n\loglog n)$ term, our algorithm is  preferable to the previous one, as we save a factor of at least $\davg^{p-1}$ for each $p$-star in $\mD(H)$.

Moreover, the  complexity of our sampling  algorithm is  parameterized by the \emph{actual} counts of the basic components $ O_{k_1}, \mydots,  O_{k_q},  S_{p_1}, \mydots, S_{p_\ell}$ of the graph $G$ at hand, rather than by the maximal possible counts of these components, respectively $m^{k_1/2}, \ldots m^{k_q/2}, m^{p_1}, \ldots m^{p_{\ell}}$, as is in  previous algorithms.
For example, if the max component cost is due to the odd cycle of length $k_1$, we get 
\[
O^*\left(\frac{m^{k_1/2} \cdot \cn o_{k_2} \cdot  \mydots \cdot \cn o_{k_q}\cdot \cn s_{p_1}\cdot \mydots \cdot \cn s_{p_{\ell}}}{\cn h}
\right)  \;\;\;\;\text{ vs. }\;\;\;\; O^*\left(\frac{m^{k_1/2} \cdot  m^{k_2/2} \mydots \cdot m^{k_q/2}\cdot  m^{p_1}\cdot \mydots \cdot m^{p_{\ell}}}{\cn h}
\right)
\]
of the previous algorithms.
Importantly, this parameterization arises \emph{only} in the analysis, while the algorithm itself is  very simple, and does not depend on prior knowledge of the actual values of these counts.

\textbf{Improved results for various graph classes.}
Our  parameterization immediately implies improved results in various interesting graph classes. 
For example, for sparse Erd\H{o}s-R\'e{n}yi random graphs $\mG(n,d/n)$, the expected count of $k$-odd cycles is  $\Theta(d^k)$, and of $p$-stars is $\Theta(n\cdot d^p$).
Hence, if we consider for example a motif $H$ that is composed of a triangle connected to a 5-petals star, our algorithm has expected complexity 
$O^*\left(\frac{m^{2.5}\cdot d^{4}}{\cn h} \right)$, while the  algorithms in~\cite{AKK19, Peng} have expected complexity $O(\frac{m^{6.5}}{\cn h})$.
In another example, for graphs of bounded arboricity\footnote{The arboricity of a graph $G$ is the minimal number of forests required to cover the edge set of $G$.} $\alpha$, the number of $k$-odd cycles is upper bounded\footnote{	
	In a graph $G$ with arboricity $\alpha$ there exists an acyclic ordering of the graph's vertices, such that each vertex has $O(\alpha)$ vertices exceeding it in the order. We can attribute each $k$-cycles in the graph to its first vertex in that ordering. It then holds that each vertex has at most $(d^+(v))^2 \cdot m^{(k-3)/2}$ attributed cycles, and it follows that $\cn O_k \leq \alpha \cdot m^{(k-1)/2}$, where $d^+(v)$ is the number of neighbors of $v$ that exceed it in the aforementioned ordering.} by $\alpha\cdot m^{(k-1)/2}$.
%Thus, instead of having an $m^{k/2}$ term for every $O_k$ in $\mD(H)$, we get  a dependence in $\alpha\cdot m^{(k-1)/2}$.
Therefore, in the case that $G$ has, e.g., constant arboricity, we save a multiplicative factor of $\sqrt m^{q}$ or $\sqrt m^{q-1}$, depending on whether the max cost component is due to  a star or an odd cycle, respectively (recall that $q$ is the number of odd cycles in the decomposition).
\iffalse
In general, for  any motif $H$ that is not an odd-cycle or a clique, 
and for any graph $G$ in which the basic counts are not maximized,  
our bounds are strictly preferable (up to lower order terms).\footnote{Note that the cycles counts  are only maximized in graphs that  contain a $K_{\Theta(\sqrt m), \Theta(\sqrt m)}$ subgraph,  and the star counts are only maximized in graphs  that contain
a vertex with degree $m$.} 
\fi

\talyaOld{Finally, while our bound might be hard to parse, we prove that indeed, if all the information we have on $H$ is its decomposition, then there exists an almost matching lower bound, up to a single $\sqrt m$ term. 
}\Ttodo{I'm not sure about this $\sqrt m$ term}

\iffalse
It is interesting to study whether the upper bound on the maximal component count for some component $R$ in $\mD$,  
can also be used to obtain a better sampling complexity for that component.
E.g., when given a bound $\cn O_3 \leq m\alpha$, can we achieve a better sampling complexity than $m^{3/2}/\cn O_3 $?

In~\cite{ERR}, Eden, Rosenbaum and Ron state (based on unpublished works) that there exists graphs $G$ with constant arboricity (so that $\cnt(G,\mC_3)$ is bounded by $O(n)$) and  $|\mC_3| =\theta(n)$,  in which sampling a  a triangle uniformly at random requires $\Omega(n^{1/4})=\omega(\cnt(G,\mC_3)/|\mC_3|)$ queries.
This suggests that only an upper bound on $\cnt(G,R)$ for some $R$ is not sufficient to reduce the complexity of the sampler all the way to $\Theta(\cnt(G, R)/|\mathcal{R}|)$.
It remains open to characterize what parameters of the graph $G$  can be helpful in reducing the sampling cost of small components (e.g., a bound on $\dmax$ is such a parameter). 

\fi

	%\input{170-results_lb}
	\subsubsection{Lower bound for estimating and sampling general motifs}
In
\ifnum\fullversion=1 
Section~\ref{sec:lb},
\else
the full version,
\fi
 we prove the following lower bound, which states that for {every} decomposition $D$ that contains at least one odd cycle component and 
every realizable  value of $\dc$,
%and a set of \emph{good} counts (where the notion of good counts will be defined shortly), 
there exists a motif $H_D$ such that $D$ is an optimal decomposition of $H_D$, and for which our upper bound is optimal.

\begin{restatable}{theorem}{dcLB}
	\label{thm:lb_dc}
	For any decomposition $D$ that contains at least one odd cycle, 
% component and has an unique minimum length odd cycle,\footnote{That is, letting $k_1$ be a minimum length odd cycle in $D$, for every $k_i\neq k_1$, it holds that $k_i > k_1$.}
%	and contains only one instance of its minimum length odd cycle component,
	and for every $n$ and $m$ and realizable value $\textproc{dc}$ of $\dc$, there exists a motif $H_D$, with optimal decomposition  $D$, and a family of graphs $\mG$ over $n$ vertices and $m$ edges, for which the  following holds.
	For every $G\in \mG$, $\dc(G,H_D,D)=\textproc{dc}$, and the expected query complexity of sampling (whp) a uniformly distributed copy of $H_D$ in a uniformly chosen $G\in G$ is $\Omega(\textproc{dc})$.
\end{restatable}

Prior to this work, the only known lower bounds  for the tasks of uniformly sampling or approximately counting motifs  $H$ that were either a clique~\cite{Eden2018-approx},  a single odd cycle~\cite{AKK19}, or a single star~\cite{GRS11, counting_stars_edge_sampling, Eden2018-approx}.
The above theorem provides  the first lower bounds for motifs with  non-trivial decompositions.
Furthermore,  even though our bounds are only \emph{decomposition-optimal} (that is, they do not hold for \emph{any} motif $H$),
each decomposition $ D$ corresponds to at least one motif $H_{\mathcal D}$ (generally, there are multiple valid ones),
for which our bounds are tight.

In order to prove Theorem~\ref{thm:lb_dc}, we actually prove a  stronger theorem, which relies on a technical notion of \emph{good counts}, formally stated in 
\ifnum\fullversion=1
Definition~\ref{def:good_cnts}.
\else
Definition  17 in the full version.
\fi

	\ifnum\fullversion=1
\begin{restatable}{theorem}{mainLB}
	\label{thm:main_lb}
	\sloppy
	For any decomposition  $D=\decomp=\{C_i\}_{i\in r}$ that contains at least one odd cycle component, for every $n,m$, $\cn h$ 	and a set of good counts,  $\{\cn c_i\}_{i\in [r]} =\{\cn o_{k_1}, \compldots, \cn o_{k_q}, \cn s_{p_1}, \compldots , \cn s_{p_\ell}\}$, as defined in 
	Definition~\ref{def:good_cnts}, 
	the following holds.  There exists a motif $H_{D}$, with an optimal decomposition $D$,  and a family of graphs $\mG$ over $n$ vertices and $m$ edges, as follows.
		For every $G\in \mG$, the basic components counts are as specified by  $\{\cn c_i\}_{i\in [r]}$, the number of copies of $H_D$ is $\cn h$, and  the expected query complexity of sampling (whp) a uniformly distributed copy of $H_D$ in a uniformly chosen $G\in \mG$ is 
	\[
	\Omega\left( \min\left\{\max_{i\in[r]}\{cost(C_i)\}\cdot \frac{\prod_{i} \cn c_{i}}{\cn h},m\right\}\right). 
	\]	
\end{restatable}

\else
\begin{restatable}{theorem}{mainLB}
	\label{thm:main_lb}
	\sloppy
	For any decomposition  $D=\decomp=\{C_i\}_{i\in r}$ that contains at least one odd cycle component, for every $n,m$, $\cn h$ 	and a set of good counts,  $\{\cn c_i\}_{i\in [r]} =\{\cn o_{k_1}, \compldots, \cn o_{k_q}, \cn s_{p_1}, \compldots , \cn s_{p_\ell}\}$, as defined in 
	Definition 17 of the full version,
	the following holds.  There exists a motif $H_{D}$, with an optimal decomposition $D$,  and a family of graphs $\mG$ over $n$ vertices and $m$ edges, as follows.
	For every $G\in \mG$, the basic components counts are as specified by  $\{\cn c_i\}_{i\in [r]}$, the number of copies of $H_D$ is $\cn h$, and  the expected query complexity of sampling (whp) a uniformly distributed copy of $H_D$ in a uniformly chosen $G\in \mG$ is 
	\[
	\Omega\left( \min\left\{\max_{i\in[r]}\{cost(C_i)\}\cdot \frac{\prod_{i} \cn c_{i}}{\cn h},m\right\}\right). 
	\]	
\end{restatable}

	\fi

\ifnum\fullversion=1
In Section~\ref{sec:dc_main},  we
\else
In the full version, we first 
\fi
 prove that  Theorem~\ref{thm:lb_dc} follows from Theorem~\ref{thm:main_lb}.
Theorem~\ref{thm:main_lb} is essentially a substantial refinement of Theorem~\ref{thm:lb_dc}, in the following sense.
Not only that for any decomposition cost we can match the lower bound (as stated in Theorem~\ref{thm:lb_dc}), but we can match it for a large variety of \emph{specific}  setting of the basic counts (as long as they are good, as stated in Theorem~\ref{thm:main_lb}). 
While Theorem~\ref{thm:main_lb} does not state that the lower bound holds for \emph{any} setting of the counts $\{\cn c_i\}_{i\in [r]}$, as we discuss in Section~\ref{sec:good_cnts}, 
some of the constraints on these counts (detailed in Definition~\ref{def:good_cnts}) are unavoidable. It remains an open question whether this set of constraints can be weakened, or perhaps more interestingly, whether, given that a set of constraints that is \emph{not good}, can a better upper bound  be devised.

	\subsection{Organization of the paper}
We give some preliminaries in Section~\ref{sec:prel}. 	
The discussion on additional related works on sublinear motif counting and sampling is deferred  to Appendix~\ref{sec:related-work}.
In Section~\ref{sec:overview_of_our_results_and_techniques} we give a high level overview of our techniques.
We present our algorithms for uniformly sampling stars and arbitrary motifs $H$ in Section~\ref{sec:general_upper_bounds}.
Due to page limitation, the full details of the $\ell_p$-sampler, approximation algorithm, as well as the 
 decomposition-optimal lower bounds  are deferred to the full version of this paper.
%
%
%and in Section~\ref{sec:lb}, we present our decomposition-optimal lower bounds  for decompositions that contain at least one odd cycle.

	\section{Preliminaries and Notation}\label{sec:prel}
Let $G=(V,E)$ be a simple undirected graph. We let $n$ denote the number of vertices in the graph.
We think of every edge $\{u,v\}$ in the graph as two \emph{oriented} edges $(u,v)$ and $(v,u)$,
and slightly abuse notation to let $m$ denote the number of oriented edges, so that $m=\sum_{v \in V} d(v)=2|E|$, and $\davg=m/n$.
Unless explicitly stated otherwise,  when we say ``edge'' we mean an oriented edge.
We let $d(v)$ denote the degree of a given vertex. We let $[r]$ denote the set of integers $1$ through $r$.

\textbf{The augmented query model.} We consider the augmented query model which allows for the following  queries. (1) A degree query, $deg(v)$,  returns the degree of $v$,  $d(v)$; (2)
An  $i\th$ neighbor query, $Nbr(v,i)$ returns the $i\th$ neighbor of $v$ if $i\leq d(v)$, and otherwise returns FAIL; (3) A pair query, $pair(u,v)$, returns whether $(u,v)\in E$; and (4) Uniform edge query returns a uniformly distributed (oriented) edge in $E$.

\textbf{A decomposition into odd cycles and stars.}
Given a motif $H$, the result in~\cite{AKK19} is parameterized by the \emph{fractional edge cover number} $\rho(H)$. The fractional edge cover number
is the optimal solution to the \emph{linear programming relaxation}
of the integer linear program (ILP) for the minimum edge cover of $H$:
The ILP allows each edge to take values in $\{ 0,1\}$, under the constraint that the sum of edge values incident to any vertex $v$ is at least $1$.
The LP relaxation allows values in $[0, 1]$ instead, and $\rho(H)$ is the minimum possible sum of all the (fractional) values.
In~\cite{AKK19}, the authors strengthen an existing result by Atserias, Grohe nd Marx~\cite{AGM},
in order to prove that there always exists an optimal solution as follows.
All of the weight (i.e., non zero edges) is supported on (the edges of) vertex-disjoint odd cycles and stars,
where each odd cycle edge has weight $\nicefrac 12$, and each star edge has weight $1$.
Consequently, the corresponding optimal solution of  the LP for a given graph $H$ 
is equivalent to a decomposition of $H$ into a collection of vertex-disjoint odd cycles and stars, denoted $\mD(H)=\decomp$. See Figure~1 for an illustration.

Generally, the motif we aim to sample (or approximate its counts) will be denoted by $H$, and the corresponding decomposition will be $\mathcal D(H)=\decomp=\{C_i\}_{i\in r}$ for $r=q+\ell$.
We use a convention of using $O_{k_i}$ to refer to the $i^{th}$ decomposition component which is an odd cycle of size $k_i$,
and $S_{p_i}$ to refer  to the $i^{th}$ star component, which is a star with $p_i$ petals.
We  use $\cn O_k$ and $\cn S_p$ denote the number of $k$-cycles and $p$-stars in $G$ respectively,
and we use $\cn h$  to denote the number of copies of $H$ in $G$.
%For the sake of ease of exposition, we define the number of stars incident to a vertex $v$ as $d^s(v)$, rather than $\binom{d(v)}{s}$.

Next, we formally define the fractional edge cover of a graph (or motif), and the resulting decomposition. We note that in this paper we will be interested in the decomposition of the motif $H$, and not the graph $G$.
\begin{definition}[Fractional edge cover]
A fractional edge cover of a graph is a function $f:E \rightarrow \mathbb{R}_{\geq 0}$ such that for every $v\in V$, $ \sum_{e \ni v} f(e) \geq 1$.
We say that the cost of a given edge cover $f$ is $\sum_{e \in E } f(e)$.
For any graph (motif) $H$, its fractional edge cover value is the minimum cost over all of its fractional edge covers,
and we denote this value by $\rho(H)$. An \emph{optimal} edge-cover of $H$ is any edge cover of $H$ with cost $\rho(H)$.
\end{definition}

\begin{lemma}[Lemma 4 in~\cite{AKK19}]\label{lem:AKK}
Any graph (motif) $H$ admits an optimal fractional edge cover $\xs$, whose support, denoted $SUPP(\xs)$,
is a collection of vertex-disjoint odd cycles and stars, such that:
\begin{itemize}
	\item for every odd cycle $C\in SUPP(\xs)$, for every $e \in C$, $\xs(e)=1/2.$
	\item for every $e \in SUPP(\xs)$ that does not belong to an odd cycle, $\xs(e)=1.$
\end{itemize}
\end{lemma}

\begin{definition}[Decomposition into odd-cycles and stars]
\label{def:decomposition}
Given an optimal fractional edge-cover $\xs$ as in Lemma~\ref{lem:AKK},
let $\{O_{k_1}, \mydots, O_{k_q}\}$ be the odd-cycles in the support of $\xs$, and let $\{S_{p_1}, \mydots, S_{p_\ell} \}$ be the stars.
We refer to $\mD(H):=\decomp$ as an (optimal) \emph{decomposition of $H$}.
\end{definition}

Given a graph (motif) $H$, its fractional edge cover value and an optimal decomposition can be computed efficiently:

\begin{theorem}[Lemma 4 and Section 3 in~\cite{AKK19}]
For any graph $H$, its fractional edge cover value $\rho(H)$ and an optimal decomposition $\mD(H)$ can be computed in polynomial time in $|H|$.
\end{theorem}

	\section{Overview of Our Results and Techniques}
\label{sec:overview_of_our_results_and_techniques}

We start with describing the ideas behind our upper bound result.

\subsection{An algorithm for sampling arbitrary motifs}
\label{sec:an_algorithm_for_general_graphs}
We take the same approach as that of~\cite{Peng}, of sampling towards estimating,
but improve on the query complexity of their bound using two ingredients.
The first is an improved star sampler, and the second is an improved sampling approach.

\textbf{Improved star sampler.} \label{sec:star-sampler}
The algorithm of~\cite{Peng} tries to sample $p$-stars by sampling $p$ edges uniformly at random,
and checking if they form a star (by simply checking if all $p$ edges agree on their first endpoint).
Hence, each $p$-star is sampled with probability $1/m^p$.
Our first observation is that it is more efficient to sample a \emph{single} edge $(u,v)$ and then sample $p-1$ neighbors of $v$ uniformly at random,
by drawing $(p-1)$ indices $i_1, \ldots, i_p$ in $[d(v)]$ uniformly at random, and performing neighbor queries $(v,i_j)$ for every $j\in [p-1]$.
However, this sampling procedure introduces biasing towards stars that are incident to lower degree endpoints.
If we were also given an upper bound $\du$ on the maximal degree in the graph, i.e., a value $\du$ such that $\dmax\leq \du$, where $\dmax$ is the maximum degree in $G$, then we could overcome the above biasing,
by ``unifying'' all the degrees in the graph to $\du$.
Specifically, this unification of degrees is achieved by querying the $i\th$ neighbor of a vertex, where $i$ is chosen uniformly at random in  $[\du]$,
rather than in $[d(v)]$.\footnote{This is effectively equivalent to rejection sampling where first $v$ is ``kept'' with probability $d(v)/\du$,
and then a neighbor of $v$ is sampled uniformly at random.}
By repeating this process $p-1$ times, we get that each specific copy of a $p$-star is sampled with equal probability $\frac{1}{m\cdot \left( \du\right)^{p-1}}$. Observe that this is always preferable to  $1/m^p$, i.e. $\frac{1}{m\cdot \left( \du\right)^{p-1}}> \frac{1}{m_p}$,  since for every graph $G$, $\du< m$. 
While we are not given such a bound on the maximal degree, letting $\cn S_p$ denote the number of $p$-stars in $G$,
it always holds that $d_{max}\leq \min\{n,\cn S_p^{1/p}\}$ (since every vertex with degree $d>p$ contributes $d^p$ to $\cn S_p$).
Hence, we can use the existing algorithms for star approximations by~\cite{GRS11, counting_stars_edge_sampling, ERS19-sidma} in order to first get an  estimate $\chat s_p$  of $\cn S_p$,
and then use this estimate to get an upper bound $\du$ on $\dmax$ by setting $\du = \min\{n,\chat S_p^{1/p}\}$.

\textbf{An improved sampling approach.}
In order to describe the  second ingredient for improving over the bounds of \cite{Peng},
we first recall their algorithm.
In the first step, their algorithm  simultaneously attempts to sample a copy 
of each odd cycle and star in the decomposition of $H$. Then if all  individual sampling attempt succeed, the algorithm proceeds to check if 
the sampled copies are connected in $G$ in a way that is consistent with the non-decomposition edges of $H$. 
However, it is easy to see that this approach is wasteful. Even if all but one of the simultaneous   sampling attempts of the first step succeed, the algorithm starts over. For example, if $\mD(H)$ consists of a star and a triangle, then in the first step their algorithm attempts to sample simultaneously  a star and a triangle, and  in the case that, say, a triangle is sampled but the star sampling attempt fails, then the sampled triangle is discarded, and the algorithm goes back to the beginning of the first step.

To remedy this, in the first step our algorithm  invokes the star- and odd-cycle samplers for every basic component in $\mD(H)$,  until all  samplers return an  \emph{actual} copy of 
 of the requested  component.  This ensures that we proceed to the next step of verifying $H$ only once we have actual copies of all the basic components.
 We then continue to check if these copies can be extended to a copy of $H$ in $G$, as before.
While this is a subtle change, it is exactly what allows us to replace the dependency in the maximum number of potential copies of the basic components, to a dependency in the actual number of copies in $G$.

We note that for motifs $H$ whose decomposition has repeating smaller sub-motifs, our sampling approach can be used recursively,
which can be more efficient.
That is, instead of decomposing $H$ to its most basic components, stars and odd-cycles, we can consider decomposing it to collections of more complex  components.
For example, if $H$ has such a collection $H_1 \subset H$ that is repeated more than once,
then it is more beneficial to first try and sample all of the copies of $H_1$
(as well as the other components of $H$) and only then try to extend these copies to $H$.
The sampling of the $H_1$ copies can then be performed by a recursive call to the motif sampler.
It can be shown that for any repeated motif $H_1$ in the decomposition of $H$,
applying the recursive sampling process results in an improved upper bound. 

\paragraph{From sampling to estimating}
In order to obtain a $(1\pm\eps)$-estimate of $\HH$, we can use the sampling algorithm as follows.
Consider a single sampling attempt in which we first sample all basic components of $\mD(H)$ (at some cost $Q$), and then preform all pair queries between the components to check if the sampled components induce a copy of $H$ (at cost $O(|H|^2)$).
By the above description such an attempt succeeds with probability that depends on the counts of the basic components of $\mD(H)$ and on the count $\HH$. Hence we can think of the success probability of each attempt as a coin toss with bias $p$, where $p$ depends only on the counts of the components and $\HH$. By standard concentration bounds, using $\Theta(1/(p\eps^2))$ sampling attempts, we can compute a $(1\pm\eps)$-estimate $\hat p$ of $p$. Since we can also get $(1\pm\eps)$-multiplicative estimates of the counts of each basic component without asymptotically increasing the running time, we can deduce from $\hat p$ a $(1\pm\Theta(\eps))$-estimate of $\HH$. See
\ifnum\fullversion=1 
 Section~\ref{sec:from_sampling_to_estimating} 
 \else
 the full version
 \fi
 for more details.

	\subsection{Decomposition-optimal lower bounds}\label{sec:overview_lb}

\textbf{Theorem~\ref{thm:lb_dc} follows from Theorem~\ref{thm:main_lb}.}
In order to prove Theorem~\ref{thm:lb_dc}, we first prove Theorem~\ref{thm:main_lb} (in Section~\ref{sec:lb}), and then prove that Theorem~\ref{thm:lb_dc} follows from Theorem~\ref{thm:main_lb} (in Section~\ref{sec:dc_main}). 
We first explain the  intuition  as to why Theorem~\ref{thm:lb_dc} follows from Theorem~\ref{thm:main_lb}.

At a high level, Theorem~\ref{thm:main_lb} states that  given (1) a decomposition $D$ and (2) a set of good counts $\{\cn c_i\}_{i\in [r]}$,  we  can construct (3) a motif $H_D$ (such that $D$ is an optimal decomposition of $H_D$) and  (4) a  family of graphs  $\mG$ such that expected number of queries required to sampling copies of $H_D$ in $\mG$  is 
\[
\comp \;.
\]
%as specified in Theorem~\ref{thm:main_lb}. 
Theorem~\ref{thm:lb_dc} states that given (a) a decomposition $D$ and (b) a (realizable) decomposition cost $\textproc{dc}$, that there exists (c) a motif $H_D$ and (d) a family of graphs for which the decomposition-cost of $G, D$ and $H_D$ is $\textproc{dc}$, and sampling copies of $H_D$ in graphs of $\mG$ requires $\Omega(\textproc{dc})$ queries.

To prove that Theorem~\ref{thm:lb_dc} follows from Theorem~\ref{thm:main_lb}, we then prove that given (a) and (b), we can specify  a set of counts which both satisfies $\textproc{dc}=\comp$ and which is good. Since the set of counts is good, we can invoke Theorem~\ref{thm:main_lb}, and get that there exists a motif $H_D$ and a family of graphs in which it is hard to sample copies of $H_D$. 
We formalize this argument in Lemma~\ref{lem:main_to_dc}, and in the rest of the section we focus our attention on the proof of Theorem~\ref{thm:main_lb}. 

\textbf{Ideas behind the proof of  Theorem~\ref{thm:main_lb}.}
Given a graph decomposition $D$, values $n$, $m$, $\cn h$ and a set of counts $\cn c_1, \mydots, \cn c_r$ of its basic components, our lower bound proof starts by defining a motif $H_{D}$, and a family of graphs $\mG$ such that  the following holds.
\sloppy
\begin{itemize}
	\item The optimal  decomposition of $H_{D}$ is $D$;
	\item For every $G\in \mG$ and $O_{k_i}, S_{p_j}\in D$, their number of copies  in $G$ is  $\Theta(\cn o_{k_i})$ and $\Theta(\cn s_{p_j})$, respectively;
	\item The number of copies of $H$ in $G$ is $\Theta(\cn h)$
	\item Sampling a uniformly distributed copy of $H_D$ in  a uniformly chosen $G$ in $\mG$, requires $\Omega\left(\min\left\{m,\textproc{dc}\right\} \right)$ queries in expectation. 
\end{itemize}

There are several challenges in proving our lower bound. First, as they are very general and work for any given decomposition $D$ that contains at least one odd cycle, there are many sub cases that need to be dealt with separately, depending on the mixture  of components in $D$.
Second, the lower bound term  does not only depend on the different counts, but also on the relations between them, which determines the component that maximizes $cost(C_i)$. As mentioned previously, our lower bound only holds for the case that the max cost is due to an odd cycle component. It remains an open question whether a similar lower bound can be proven for the case that the max cost is due to a star, or whether in that case a better algorithm exists. The authors suspect the latter option.
Third, as in most previous lower bounds for motif sampling and counting, we prove the hardness of the task by ``hiding'' a constant fraction of the copies of $H_D$, so that the existence of these copies depends on a small set of crucial edges. That is, we prove that we can construct the family of graphs $\mG$,  such that for every $G\in\mG$,  a specific set of $t$ crucial edges, for some small $t$ that depends on the basic counts and $\cn h$, contributes $\Theta(\cn h)$ copies of $H_D$ . We then prove that detecting these edges requires many queries (this is formalized by a reduction from a variant of the \SD\ communication complexity problem, based on the framework of~\cite{Eden2018-approx}).
This approach  of constructing many copies of $H_D$ which all depend on  small set of crucial edges, leads the construction of the graphs $\mG$ to contain  very dense components, which in turn causes correlations between the counts of the different components. A significant challenge  is therefore to define the motif $H_D$ and the graphs of $\mG$ in  a way that satisfies all given counts simultaneously.

In each graph  $G$ in the hard family $\mG$,  we have a corresponding ``gadget'' to each of the components of  $D$. 
Let $k_1$ denote (one of) the maximum-cost odd-cycle components.
For each odd-cycle component $O_{k_i}$ for $k_i\neq k_1$, we define either a \SCG\ or a
\MM{} that induce $\cn o_{k_i}$ odd cycles of length $k_i$ according to the relation between $k_i$ and $k_1$.
For each star component $S_{p_j}$ we define a \SSS{} that induces $\cn s_{p_j}$ many $p_j$-stars.
The maximum-cost cycle component $O_{k_1}$ has a different gadget, a \CC{}. This gadget is used to hide the set of $t$ crucial edges, and allows us to parameterize the complexity in terms of the cost $cost\{O_{k_1}\}$. 

To formally prove  the lower bound we  make use the framework introduced in~\cite{Eden2018-approx}, which uses reductions from communication complexity problems to motif sampling and counting problems in order to prove hardness results of these latter tasks. 
This allows us to prove that one cannot, with high  probability, witness an edge from the set of $t$ hidden edges, unless $\Omega(m/t)$ queries are performed. This in turn implies that one cannot, with high probability, witness a copy of $H_D$ contributed by these edges.
Hence, we obtain a lower of $\Omega(m/t)$ for the task of outputting a uniformly sampling. Setting $t$ appropriately gives the desired bound.

	\renewcommand{\rmax}{r_{max}}

\section{Upper Bounds for Sampling Arbitrary Motifs}
\label{sec:general_upper_bounds}

In this section we present our improved sampling algorithm.
Recall that our upper bound improvement has two ingredients, an improved star sampler, and an improved sampling approach. 
We start with presenting the improved star sampling algorithm.

\subsection{An optimal ($\ell_p$-)  star-sampler} 

Our star sampling procedure assumes that it gets as a parameter  
a value $\hSp$ which is a constant-factor estimate of $\Sp$. This value can be obtained by invoking one of the star estimation algorithm of~\cite{counting_stars_edge_sampling,ERS19-sidma}.

\begin{lemma}[\cite{counting_stars_edge_sampling}, Theorem 1]\label{lem:star-est-ERS}
	\sloppy
	Given query access to a graph $G$ and an approximation  parameter $\eps$, 
	there exists an algorithm, \textsf{Moment-Estimator}, that returns  a value $\hSp$,  such that with probability at least $2/3$,
	$\hSp\in [\cn s_p,2\cn s_p]$. The expected query complexity and running time $O\left(\min\left\{m,\min\left\{\frac{m\cdot n^{p-1}}{\Sp},\frac{m}{\Sp^{1/p}}\right\} \cdot {\loglog n}\right\} \right)$. 
\end{lemma}

Given an estimate $\hSp$ on $\Sp$, our algorithm sets an upper bound\footnote{Observe that $\dmax$ is $\dmax=\max_{v}d(v)$, while $\du$ is simply a bound on $\dmax$, so that $\dmax\leq \du$.} $\du$ on the maximal degree, $\du=\min\{n,\hSp\}$. It then tries to sample a copy of a $p$-star as follows. In each sampling attempt it  
 samples a single edge $(v_0,v_1)$, and then performs $p-1$ neighbor queries $nbr(v_0,i_j)$ for $j=2\ldots p$, where each $i_j$ is chosen independently and uniformly at random from $[\du]$.  In order to ensure that the sampled neighbors are distinct, and to avoid multiplicity issues, a $p$-star is returned only if its petals are sampled in ascending order of ids.
In  every such sampling attempt, each specific $p$-star  is therefore sampled with equal probability $\frac{1}{m\cdot \du^{p-1}}$. Hence, invoking the above $\frac{m\cdot \du^{p-1}}{\Sp}$ times, in expectation, returns a uniformly distributed copy of a $p$-star.

\alg{
	{\SaS$(p,n, \hSp)$} \label{alg:sample_star}\label{sas}
	\smallskip
	\begin{compactenum}
		\item Let $\du=\min\{n, (c_p\cdot \hSp)^{1/p}\}$ for a value $c_p$ as specified in the proof of Theorem~\ref{lem:star-sampler}. \label{step:set-du}
		\item While \textbf{TRUE}:
		\begin{compactenum}
			\item Perform a uniform edge query, an denote the returned  edge $(v_0,v_1)$.\label{step:samp-edge}
			\item Choose $p-1$ indices $i_2, \ldots, i_{p}$ uniformly at random in  $[\du]$ (with replacement).  \label{step:samp-nbrs}
			\item  For every $j\in [2..p]$, query the $i_j \th$ neighbor of $v_0$. Let 
			 $v_2, \ldots, v_{p}$ be the returned vertices, if all queries returned a neighbor. Otherwise break.
			 \item If $id(v_2)<id(v_2)<\ldots<id(v_p),$ then \textbf{return} $(v_0,v_1, \ldots, v_{p})$.	
%			\item If the number of queries exceeds $\hat m$, query the entire graph,\footnote{by either performing $n$ degree queries and $2m$ neighbor queries, or $m\log m$ uniform edge queries} and return a uniformly distributed copy of $S_p$. \label{step:abort}
		\end{compactenum}
	\end{compactenum}
}

\begin{theorem}
	\label{lem:star-sampler}
	\sloppy
Assume that  $\hSp\in [\cn s_p, c\cdot \cn s_p]$ for some small constants $c$. 
The procedure \SaS$(p,\hSp)$ returns a uniformly distributed $p$-star in $G$. The expected query complexity of the procedure is  $O\left(\min\left\{\frac{m\cdot n^{p-1}}{\Sp},\frac{m}{\Sp^{1/p}} \right\}\right)$.
%\cdot \frac{\log^2(n/\delta)\loglog n}{\eps^2} \right\}\right)$.
\end{theorem}
\begin{proof}

\sloppy
Let $c_p$ denote the minimal value such that for every $k\in [n]$, $c_p \cdot \binom{k}{p} \geq k^p$
%so that in particular, 
%$c_p \cdot \binom{\dmax}{p} \geq  \dmax^p$ 
(note that $c_p=\Theta(p!)$).
Then $\Sp =\sum_{v \in V}\binom{d(v)}{p}>\binom{\dmax}{p} \geq  \dmax^p/c_p$, and by the assumption on $\hSp$, $\dmax < (c_p \cdot \Sp)^{1/p} \leq (c_p\cdot \hSp)^{1/p}$. It follows by the setting of $\du=\min\{n,(c_p\cdot \hSp)^{1/p}\}$ 	in Step~\ref{step:set-du}, that $\du\geq \dmax$.

Consider a specific copy $\bar S_p=(a_0, a_1, \ldots, a_p)$ of a $p$-star in $G$, where $a_0$ is the star center and $a_1$ through $a_p$ are its petals in ascending id order. In each iteration of the while loop, the probability that $\bar S_p$ is returned is 
\begin{align*}
 \Pr[\bar S_p \text{ is returned}] &=\Pr[(a_0, a_1) \text{ is sampled in Step~\ref{step:samp-edge}}] 
%\\ & \hspace{2cm}
\cdot \Pr[a_2,\mydots, a_p \text{ are sampled in Step~\ref{step:samp-nbrs}}]\\
	&= \frac{1}{m}\cdot \frac{1}{\du^{p-1}} \;. \numberthis \label{eq:samp-prob-star}
\end{align*}
Note the the last equality crucially depends on $d(v)\leq \dmax\leq \du$ for all $v\in V$. (Indeed, if there exists a vertex $v$ with degree $d(v) > \du$, then some of its incident stars will have zero probability of being sampled.) Hence, each copy is sampled with equal probability, implying that the procedure returns a uniformly distributed copy of a $p$-star.

We now turn to bound the expected query complexity.
It follows from  Equation~\ref{eq:samp-prob-star} and the setting of $\du$, that the success probability of a single invocation of the while loop is 
$ \frac{\Sp}{m\cdot \du^{p-1}}$. Hence, 
%in the event  that ${\hSp} \in [\Sp, 1.5 \Sp]$, 
the expected number of invocations is $\frac{m\cdot \du^{p-1}}{\Sp}$. 
%In the event that ${\hSp} \notin [\Sp, 1.5 \Sp]$,  the  number of invocations is bounded by $O(n+m)=O(n^2)$ due to step~\ref{step:abort}.
%Since   ${\hSp} \in [\Sp, 1.5 \Sp]$ with probability at least $1-\delta/n^2$, 
It follows that, for a constant $p$,  the expected number of invocations is \[O\left(\frac{m\cdot \min\{n, (c_p\cdot \Sp)^{1/p} \}^{p-1}}{\Sp}\right) = O\left(\min\left\{ \frac{m\cdot n^{p-1}}{\Sp}, \frac{m}{\Sp^{1/p}} \right\}\right) .\]
Since the query complexity and running time of a single invocation of the while loop are constant, the above is also a bound on the expected query complexity and running time of the while loop. 
\end{proof}

In the full version of this paper, we explain how  algorithm \SaS\ can be slightly modified to produce an $\ell_p$-sampler, \lps\,
as specified in Theorem~\ref{thm:lpsamp}.

\ifnum\fullversion=1
\medskip
\textbf{A sublinear $\ell_p$-sampler}
The algorithm \SaS\ can be slightly modified to produce an $\ell_p$-sampler, \lps\,
so that each vertex $v$ is returned with probability $d(v)^p/\mu_p$. First we assume that we are given a value $\hat \mu_p$ which is a constant factor estimate of $\mu_p$, rather than being given an estimate $\hSp$ on $\cn s_p$ (recall that $\mu_p=\sum_{v} d(v)^p$). As discussed in~\cite{counting_stars_edge_sampling,ERS19-sidma}, this can be obtained 
with slight modifications to their algorithms.  
Given $\hat \mu_p$ such that $\hat{\mu}_p\in [\mu_p, c\cdot \mu_p]$ for some small constant $c$, the algorithm sets $\du=\min\{n, \hat \mu_p^{1/p}\}$. It then proceeds as in \SaS, to query a uniformly distributed edge $(v_0,v_1)$, and to perform $p-1$ $i\th$-neighbor queries for $p-1$ indices chosen uniformly at random (with replacement) in $[\du]$. If all neighbor queries succeed, then the algorithm returns $v_0$.
To see that every vertex is returned with probability  $d(v)^p/\mu_p$,
fix an iteration of the while loop. The probability that a vertex $v$ is returned is 
\begin{align*}
	\Pr[v \text{ is returned}] &=\frac{d(v)}{m}\cdot \left(\frac{d(v)}{\du}\right)^{p-1}
	= \frac{d(v)^{p}}{m \cdot \du^{p-1}} \;. \numberthis 
\end{align*}
Hence, the probability that any vertex is returned in a single invocation is $\frac{\mu_p}{m \cdot \du^{p-1}}$. 
Therefore, for every vertex $v$, by Bias theorem, the probability that $v$ is the returned vertex is 
\[
\frac{d(u)^p}{\mu_p}.
\]
The expected query complexity and running time of the procedure are 
$O\left(\min\left\{ \frac{m\cdot n^{p-1}}{\mu_p}, \frac{m}{\mu^{1/p}} \right\}\right).$

\fi

\subsection{General motif sampler}

Our algorithm for sampling uniform copies of a motif $H$ in a graph $G$ relies on the above star sampler, and the odd cycle sampler of~\cite{Peng}.

\begin{lemma}[Lemma 3.3 in \cite{Peng}, restated]\label{lem:peng-samp-cyc}
There exists a procedure that, given a parameter $k$ and an estimate $\hat m\in [m,2m]$	, samples each specific copy of an odd cycle of length $k$ with probability $1/m^{k/2}$.
\end{lemma}

It follows that by repeatedly invoking the procedure above until an odd cycle is returned we can get an odd cycle sampling algorithm.
\begin{corollary}\label{cor:cycle-sampler}
	There exists a procedure, \SoC,  that, given an estimate $\hat m\in [m,2m]$,  returns a uniformly distributed copy of an odd cycle of length $k$. The expected query complexity is  $O\left(\min\left\{m\log n,n+m,\frac{m^{k/2}}{\cn o_k}\right\}\right)$, where $\cn o_k$ denotes the number of odd cycles of length $k$ in $G$.
\end{corollary}

We also use the following algorithm from~\cite{GR08} to obtain an estimate of $m$.
\begin{theorem}[~\cite{GR08}, Theorem 1, restated]\label{thm:gr08}
	There exists an algorithm that, given query access to a graph $G$, the number of vertices $n$, and a parameter $\eps$, returns a value $\tilde m$, such that with probability at least $2/3$, $\tilde m\in [m,(1+\eps)m]$.
	The expected query complexity and running time of the algorithm is $O(n/\sqrt m)\cdot (\loglog n/\eps^2)$.
\end{theorem}

Our motif sampling algorithm  invokes the star-sampler and odd-cycles-sampler for each of the star and odd-cycles components in $\mD(H)$, respectively. Once actual copies of all the components are sampled, it checks whether they form a copy of $H$ in $G$, using $O(|H|^2)=O(1)$ additional pair queries.

\alg{
    {\bf \SH$\;(H,n)$} \label{alg:count-H}
    \smallskip
    \begin{compactenum}
    	\item Compute a 2-factor estimate $\hat m$ of $m$ by invoking the algorithm of~\cite{GR08} with $\eps=1/2$ for $10\log n$ times, and letting $\hat m$ be the median of the returned values.
    	\item Compute an optimal decomposition of $H$,  $\mD(H)=\decomp$. \label{step:decomp}
    	\item For every $S_{p_i}$ in $D$, invoke algorithm \textsf{Moment-Estimator}  with $\eps = 1/2$ and $r=p_i$  for $t=10\log(n\cdot \ell)$ times		
    	to get $t$ estimates of $\cn s_{p_i}$.\label{repeat-t-SH}
    	Let $\chat s_{p_i}$  be the median value among the $t$ received estimates of each $S_{p_i}$. \label{repeat-t}
 
%    	\item Let $\delta=1/(|H|\cdot m^{\rho(H)}).$
%    	\item If $\mD(H)$ is composed solely of stars, then invoke \SE\ and \textbf{return} its returned value.
		\item While \textbf{True:}
		\begin{compactenum}
         \item For every $i\in [q]$ do:
          \begin{compactenum}
                \item Invoke \SoC($k_i, \hat m$), and let $\bar O_i$ be the returned odd cycle.  \label{step:samp-cycle}
                %If no cycle was returned \fail.
            \end{compactenum}
            \item For every $i\in [\ell]$ do: 
            \begin{compactenum}
				\item Invoke \SaS($p_i,n,\chat s_{p_i}$), and let $\bar S_j$ be the returned $s_j$-star. \label{step:samp-star} 
				%If no star was returned \fail.
            \end{compactenum}
			\item Perform $O(|H|^2)$ pair queries to verify whether the set of components $\{\bar O_1, \ldots, \bar O_q, \bar S_1, \ldots, \bar S_{\ell}\}$ can be extended to a copy of $H$ in $G$.\label{step:verify-copy}
			\item If a copy of $H$ is discovered, then \textbf{return} it.\label{step:return-H} 
			\item If the number of queries performed exceeds $n+\hat m$, then query all edges of the graph\footnote{by either performing $n$ degree queries and $2m$ neighbor queries, or $10m\log n$ uniform edge queries} and output a uniformly distributed copy of $H$.\label{step:query_all}
			%Otherwise, \fail. 
    \end{compactenum} 
   \end{compactenum}
}
%{A procedure for sampling any motif $H$,  such that  each copy is sampled with equal probability $\Pi_{i=1}^{\ell} s_i!/m^{\tau(H)}$}

%We note that the above can be enhanced to succeed with probability at least $1-\delta$ for any $0<\delta<1$	by invoking the algorithm $\log(1/\delta)$ times and taking the median returned value.

We are now ready to prove our main upper bound theorem, which we recall here.
\ubMain*

%\begin{theorem}
%Let $G$ be a graph over $n$ vertices, $m$ edges, and let $H$ be  a motif such that $\mD(H)=\decomp$, where the $i\th$ cycle is of length $k_i$ and the $i\th$ star has $p_i$ petals. With probability at least $1-\delta$, Algorithm \SH\ returns a uniformly distributed copy of $H$. The expected query complexity of the algorithm is 
%\[
%%	O\left( \min\left\{\max_{i\in[r]}\left\{cost(C_i) \right\}  \cdot \frac{\prod \cn c_{i} }{\HH},m\right\} \cdot \log n\loglog n\right).
%O\left(\min\left\{decomp-cost(G,H,\mD(H)),m \right\}\cdot \log n\loglog n\right)
%\]
%\end{theorem}

\begin{proof}

%We start with proving that, with high probability,  the algorithm returns a uniformly distributed copy of $H$. 
By Theorem~\ref{thm:gr08}, when invoked with a value $\eps=1/2$, the edge estimation algorithm of ~\cite{GR08} returns a value $\tilde m$ such that, with probability at least $2/3$, $\tilde m\in [m,1.5m]$. Hence, with  probability at least $1-1/3n^2$, the median value $\hat m$ of the $10\log n$ invocations is such  that $\hat m\in [m,1.5m]$. We henceforth condition on this event.

We next prove  that with probability at least $1-1/3n^2$, all the computed $\chat s_{p_i}$ values are good estimates of $\cn s_{p_i}$.
By Lemma~\ref{lem:star-est-ERS}, for a fixed $p_i$, with probability at least $2/3$, the value returned from \textsf{Moment-Estimator} is in $[\chat s_{p_i}, 1.5\cdot  \chat s_{p_i}]$. 
Therefore, the probability that the median value of the $t=10\log(n\ell)$ invocations in Step~\ref{repeat-t-SH} is outside this range is at most $1/(3\ell n^2)$.
Hence, taking a union bound over all $i\in[\ell]$, with probability at least $1-1/3n^2$, for every $i\in[\ell]$, $\chat s_{p_i}\in  [\cn s_p, 1.5 \cdot \cn s_p]$. We henceforth condition on this event as well.

Fix a copy $H'$ of $H$ in $G$, and let $ O'_1, \mydots,  O'_q,   S'_1, \mydots,  S'_\ell$ be its cycles and stars,  corresponding  to those of $\mD(H)$. 
By Corollary~\ref{cor:cycle-sampler}, for each $O'_i$, its probability of being returned in Step~\ref{step:samp-cycle} is $1/\cn O_{k_i}$. Similarly, by Lemma~\ref{lem:star-sampler}, for each $S'_i$, its probability of being returned in Step~\ref{step:samp-star} is $1/\cn S_{p_i}$.	
Therefore, in the case that the number of queries does not exceed $\hat m$, in every iteration of the loop, each specific copy of $H$ is returned with equal probability $\frac{1}{\Pi_{i=1}^q \cn O_{k_i}\cdot \Pi_{i=1}^\ell \cn S_{p_i}}$.~\footnote{To avoid multiplicity issues, if some components are repeated in the decomposition more than once, then we can assign ids to small components and verify they are sampled in ascending id order.} Hence, once a copy of $H$ is returned, it is uniformly distributed in $G$. In the case that the number of queries exceeds $\hat m$, the algorithm either performs $n+2m$ queries to query all the neighbors of all vertices, or $10m\log n$ queries, in order to discover all edges with high probability. In the former case, the entire graph $G$ is known. In the latter case, by the coupon collector analysis, the probability that all edges are known at the end of the process is at least $1-1/3n^2$. Hence, with probability at least $1-1/3n^2$, at the end of this process, a uniformly distributed copy of $H$ is returned.

It remains to bound the query complexity. 
By Lemma~\ref{lem:star-est-ERS}, Step~\ref{repeat-t-SH} takes $\sum_{p_i}t\cdot \min\left\{\frac{m\cdot n^{p_i-1}}{\cn s_{p_i}} , \frac{m}{\cn S_{p_i}^{1/p_i}} \right\}\cdot \log n\loglog n$ queries in expectation. 
By the above discussion, it holds that the expected number of invocations of the while loop is $\frac{\Pi_{i=1}^q \cn O_{k_i}\cdot \Pi_{i=1}^\ell \cn S_{p_i}}{\HH}$.
% and the probability that the number of invocations is $R$ times larger is exponentially small in $k$.	
Furthermore, by Lemma~\ref{lem:star-sampler}, the expected  query complexity of sampling each  $S_{p_i}$ is  $\min\left\{\frac{m\cdot n^{p_i-1}}{\cn s_{p_i}} , \frac{m}{\cn S_{p_i}^{1/p_i}} \right\}$. 
By Lemma~\ref{cor:cycle-sampler}, the expected running time of each invocation of the $k_i$-cycle sampler is $O\left(\frac{m^{k_i/2}}{\cn O_{k_i}}\right)$. 
The complexity of Step~\ref{step:verify-copy} is $O(|H|^2)=O(1)$ queries, and is subsumed by the complexity of the other steps.
Hence, the expected cost of each invocation of the while loop is 
\[\max_{i \in [q]}\left\{\frac{m^{k_i/2}}{\cn O_{k_i} } \right\} + \max_{i\in [\ell]}\left\{\min\left\{\frac{m}{\cn S_{p_i}^{1/p_i}}, \frac{m\cdot n^{p_i-1}}{\cn S_{p_i}}\right\} \right\} =
 \max_{i \in [q]}\left\{\frac{m^{k_i/2}}{\cn O_{k_i} } \right\} + \min\left\{\frac{m}{\cn S_{p}^{1/p}}, \frac{m\cdot n^{p-1}}{\cn S_{p}} \right\},\]
 where the equality holds since the maximum of the second term is always achieved by the largest star in the decomposition, $S_p.$
 % (by norm inequalities, for any $p_i\leq p$, $\cn s_{p}^{1/p}\leq \cn s_{p_i}^{1/p_i}$).
 Also, due to Step~\ref{step:query_all} and the assumption on $\hat{m}$, the query complexity of algorithm is always bounded by  $O(\min\{m\log n, n+m\})$. 
Therefore, the overall expected query complexity is the minimum between $O(\min\{m\log n, n+m\})$ and 
\begin{align*}
&O\left( \left(\max_{i\in[q]}\left\{\frac{m^{k_i/2}}{\cn O_{k_i}}\right\} + \min\left\{ \frac{m\cdot n^{p-1}}{\Sp}, \frac{m}{\Sp^{1/p}}\right\}\cdot \log n \loglog n   \right)\cdot \frac{\prod_{i\in[r]}\cn c_i}{\HH} \right)
\\& =O\left( \min\left\{\max_{i\in[r]}\left\{cost(C_i) \right\}  \cdot \frac{\prod \cn c_{i} }{\HH},m\right\} \cdot \log n\loglog n\right)
\\& =O\left(\min\left\{\dc(G,H,\mD(H)),m,n \right\}\cdot \log n\loglog n\right),
\end{align*}
as claimed.
\end{proof}

\ifnum\fullversion=1 
\subsection{From sampling to estimating}
\label{sec:from_sampling_to_estimating}
%\alg{
%	{\Est$(H)$} \label{alg:estimate}\label{alg:est}
%	\smallskip
%	\begin{compactenum}
%		\item Compute $\rho(H)$.
%		\item Let $\mu = m^{\rho(H)}$ and $\delta =\log \mu?$.
%		\item Repeat $?$ times:
%		\begin{compactenum}		
%		\item Invoke algorithm $\mA(H)$  and if it runs for more than .	
%		
%		
%		Let $\chi_{\mu}$ be the number of invocations that 
%		
%		to get $t$ estimates of $\Sp$.
%		Let $\hSp$  be the median value among the $t$ received estimates.
%		\item Let $\du=\min\{n, \hSp^{1/p}\}$. \label{step:set-du}
%		\item While \textbf{TRUE}:
%		\begin{compactenum}
%			\item Return $(v_0,v_1, \ldots, v_{p})$.	
%		\end{compactenum}
%		\end{compactenum}
%
%	\end{compactenum}
%}

\begin{theorem}\label{thm:samp_to_est}\sloppy
There exists an algorithm $\mA$ that returns an estimate $\ohest$ of $\oh$ such that with probability at least $2/3$, $\ohest\in (1\pm\eps)\oh$ and the query complexity of $\mA'$ is 
% $O\left(\min\left\{m,\comp\cdot \frac{\log n \loglog n}{\eps^2}\right\}\right)$.
$O\left(\min\left\{decomp-cost(G,H,\mD(H)),m \right\}\cdot \frac{\log n\loglog n}{\eps^2}\right)$.
\end{theorem}
\begin{proof}
    We prove the claim by describing the algorithm. First $\mA$ computes an optimal  decomposition of $H$, $\mD(H)=\decomp$. It then computes  estimates of all odd-cycles and stars in  $\mD(H)$ with $\eps'=\eps/3|H|$ and $\delta=1/6(q+\ell)$, using the algorithms of~\cite{AKK19} and~\cite{counting_stars_edge_sampling}, respectively. 
    Consider a single invocation of the inner loop of algorithm \SH. 
    By the analysis of algorithm \SH, it holds that the success probability of a single such invocation is $\frac{\HH}{\cn O_{k_1}\cdots \cn O_{k_q} \cdot \cn S_{p_1} \cdots  \cn S_{p_\ell}}$. Denote this probability by $p$.
    We can think of the above as tossing a coin with bias $p$.
    By the multiplicative Chernoff  bound, with probability at least $5/6$, the bias $p$ can be approximated up to a $(1\pm\eps/3)$-multiplicative error in $O(\frac{1}{\eps^2 \cdot p})$ tosses. Since  the individual counts $\cn O_{k_1}\ldots \cn O_{k_q}, \cn S_{p_1} \ldots  \cn S_{p_\ell}$ are also known up to a $(1\pm\eps/3|H|)$-multiplicative error, given the estimate of $p$, the algorithm can extract a $(1\pm\eps)$-estimate of $\HH$.
    Estimating the individual counts of the basic components of $\mD(H)$ takes $O\left(\max_{i\in[r]}\left\{cost(C_i) \right\}\cdot \log n\loglog n\right)$ time  in expectation.
    Each invocation of the inner loop of \SH\  takes $O\left(\max_{i\in[r]}\left\{cost(C_i) \right\}\right)$ time in expectation. 
    Finally, we can track the number of queries performed by the algorithm and in case it exceeds $m$ we can query all the edges of the graph, and simply return the count of $H$. 
    Hence, the expected query complexity of the algorithm is 
	$O\left(\min\left\{m,\comp\right\}\cdot \frac{\log n \loglog n}{\eps^2}\right)$. 
\end{proof}	

\fi
	
	%\ifnum\fullversion=1
	\section{Lower Bounds}\label{sec:lb}

In this section we prove our main lower bounds statements, Theorem~\ref{thm:lb_dc} and Theorem~\ref{thm:main_lb}.
We defer the proof   that the former follows from the latter to Section~\ref{sec:dc_main}, 
and start with proving Theorem~\ref{thm:main_lb}, stated here again for the sake of convenience. 

\mainLB*

We next formalize the definition of \emph{good counts}.

\subsection{Good counts}
\label{sec:good_cnts}

\begin{definition}[Good counts]\label{def:good_cnts}
	We say that a set of counts $n,m$ and $\cn o_{k_1}, \compldots, \cn o_{k_q}, \cn s_{p_1}, \compldots , \cn s_{p_\ell}, \cn h$ is \emph{good} if the following hold.
	\begin{enumerate}
		\item \label{const:valid}
		The counts are realizable; that is, there exist a graph $G$ and a motif $H_D$ with optimal decomposition $D$ that realize these counts.
		\item \label{const:max_cost}
		The max component cost is due to an odd cycle component. That is, $\textrm{argmax}_{i\in [r]}\{cost(C_i)\}=O_{k_i}$ for some odd cycle component $O_{k_i}\in D$.  Assume without loss of generality that $O_{k_1}$ is the odd cycle that maximizes $\max_{i\in [r]}\{cost(C_i)\}$. 
		\item \label{const:cycles}
		$\forall k_j>k_i$, if $\cn o_{k_i}\leq \sqrt m^{k_i-1}$, then , $\left( \cn O_{k_j}\right)^{1/(k_j-1)} \geq\left(\cn o_{k_i}\right)^{1/(k_i-1)}$.
		
		Otherwise, if $\cn o_{k_i}> \sqrt m^{k_i-1}$,  $\left(  \cn O_{k_j}\right)^{1/k_j}\geq \left( \cn O_{k_i}\right)^{1/k_i}$.
		\item  \label{const:stars_lb}
		For every $j\in[\ell]$, $\cn S_{p_j}\geq \sqrt m^{p_j+1}$.
		\item At least one of the followings hold.\label{const:short_or_star}
		\begin{enumerate}
	        \item 
	        Let $k_*$ be the index of the $O_k$ that maximizes $\cn o_k^{1/k}$.
	        There exists at least one star $S_p$ in $D$ with  $\cn s_p=\omega\left(m\cdot (\cn o_{k_*})^{(p+1)/k_*}\right)$.
	        Observe that it always holds that $\cn o_{k_*}^{1/k_*}\leq \sqrt m$, so if  $\cn s_p=\omega(\sqrt m^{p+3})$ then this constraint holds.
	        \label{const:one_star}
			\item \label{const:short_cycles}
			For every $k_i\leq k_1$, $\cn k_i\leq \sqrt m^{k_i-1}$.
		\end{enumerate}
		\item At least one of the followings hold.\label{const:stars}
		\begin{enumerate}
			\item \label{const:stars_a}
			For at least  one of the cycles $O_k$, it holds that  $\cn o_k \leq  \sqrt m^{k-1}$, and for every $p$, $\cn s_p\geq n^p$.   
			\item \label{const:stars_b}
			The $\cn s_{p_i}$ counts are such there exists a set $A$ of $\sqrt m$ integers $a_1, \ldots, a_{\sqrt m}$ so that $\forall i,\; a_i\leq n$, $\sum_{i} a_i\leq m$, and $\sum_{i} a_i^{p_i}=\cn s_{p_i}$.
		\end{enumerate}
		
	\end{enumerate}
\end{definition}

As discussed in the introduction, some of the above constraints are unavoidable, and some arise due to the way we construct the graphs $G$ in the hard family $\mG$. Details follow.

\begin{enumerate}
  \item Constraint~\ref{const:valid} simply states that the given counts can be realized by some graph and is therefore  unavoidable. 
  \item Constraint~\ref{const:max_cost} implies that our upper bound is tight only in the case that the max cost is due to an odd cycle and not due to a star component.
We leave it as an open question  whether  for the case that the max component cost is due to a star, a new lower bound can be designed or an improved algorithm can be devised.
\end{enumerate}
The rest of the constraints arise from the way we construct  the basic structure of the graphs in the ``hard'' family of graphs in the proof of the lower bound. 
\vspace{-.2cm}
\begin{enumerate}
  \setcounter{enumi}{2}
\item Constraint~\ref{const:cycles}: for each cycle $O_{k_i}$ such that $\cn o_{k_i\geq \sqrt m^{k_i-1}}$, we ``pack'' the $\Theta(\cn o_{k_i})$ ${k_i}$ length odd cycles in a  $k_i$-partite subgraph. This inadvertently results in the creation of $\Theta((\cn o_{k_i})^{k_j/k_i})$ odd-cycles for any $k_j\geq k_i$ length odd cycle component.
\item Constraint~\ref{const:stars_lb}: Recall that in order to prove the lower bound we ``hide'' as set of $t$  crucial edges which create  $\Theta(\cn h)$ of the copies of $H_D$. To hide the edges,  we use a subgraph with density $\Theta(\sqrt m)$, which again inadvertently induces $\Theta(\sqrt m^{p+1})$ $p$-stars for every $p\in [\sqrt m]$.
\item Constraint~\ref{const:short_or_star}: 
Let $k'$ denote the min length odd cycle component in $D$. If for example $\cn o_{k'}=\sqrt m^{k'}$, then our gadget for creating  $\cn o_k$ odd cycles also maximizes (up to constant factors) the counts of all odd cycles for every $k_i$, and therefore might induce too many copies of $H_D$.
To avoid such a scenario, we require that either there exists  at least one star in $D$ with counts strictly greater than what could be created by a cycle gadget (in~\ref{const:one_star}); or that the number of short cycles, i.e., cycles of length $k_i\leq k_1$, does not exceed $\sqrt m^{k_i-1}$ (in~\ref{const:short_cycles}). In the latter case the corresponding gadget can have a single vertex which is incident to all cycles, and therefore, no two vertex-disjoint odd cycles can be formed, so that no copies of $H_D$ are formed solely by this gadget.
\item Constraint~\ref{const:stars} arises from the way we connect the odd cycles and stars in the graphs of $\mG$. The first item, ~\ref{const:stars_a}, simply states that the count of one of the cycles which is not the max cost cycle is not maximized. 
In such a case the corresponding cycle gadget will have one part with a single vertex, which will allow us to connect it to a set of $n$ vertices that induce the $\cn s_p$ counts in the corresponding star gadget.
The second item, item~\ref{const:stars_b},  states that there exists a set $A$ of $|A|\leq \sqrt m$ (rather than $n$) integers (that will later determine the degrees of $|A|$ vertices), so that for every $p$, $\sum_{a_i\in A} a_i^p=\cn s_p$.~\footnote{Note that indeed there exists many \emph{valid} counts (ones which can be realized by some graph) that satisfy this constraint. Consider first a  bipartite graph $G_0=A\cup B$ with $|A|=\sqrt m$, $|B|=n$, where each vertex in $A$ has degree $\Theta(\sqrt m)$, and each vertex in $B$ has degree $O(\sqrt m)$. Then in this graph, all star counts are exactly $\cn s_p=\sqrt m^{p+1}$ as required by the second constraint. To get higher values of the counts $\cn  s_p$, we can simply move edges around, one edge at a time, as to  skew  the set of degrees of the vertices of $A$. 
	Let $G_t$ denote the graph resulting from the above process at time $t$. This process ends after $r$ steps, with a graph  $G_r=A'\cup B'$ as follows.  $A'$ has $\davg$ vertices with degree $n$, and $\sqrt m-\davg$ vertices of degree $0$, and $B$ has $n$ vertices with degree $\davg$. This graph maximizes the $\cn s_p$ counts, $\cn s_p=\davg\cdot n^p$ for any $p$. At each time step $t$, the set of   counts $\cn s_{p_1}, \ldots, \cn s_{p_\ell}$ of the $p_i$-stars in $G_t$ satisfies constraint~\ref{const:stars_b}.}
\end{enumerate}

We note that while there are indeed many constraints required by our construction, 
these constraints are satisfiable by many sets of possible counts. Indeed in order prove that  Theorem~\ref{thm:lb_dc} follows from Theorem~\ref{thm:main_lb} (see proof of Lemma~\ref{lem:main_to_dc}), we show that for \emph{every} realizable value of $\dc(G,H,\mD(H))$, 
\emph{there exists} 
a set a set of good  counts $\{\cn c_i\}_{i\in [r]}$, 
which satisfies all of the constraints of Definition~\ref{def:good_cnts}.

We continue to describe the different ingredients required for our proof.
% To prove the theorem, we start by defining a motif $H_D$ such that $D$ is an optimal decomposition of $H_D$.  We then describe 
% a``hard'' family of graphs $\mG$, and prove that sampling uniformly distributed copies of $H_D$ in a uniformly chosen  graph in $\mG$ is as specified. 
We make use of the framework  for proving graph estimation  lower bounds via communication complexity reductions given  in~\cite{Eden2018-approx}.
The framework makes use of the following communication problem.
\begin{theorem}
	\sloppy
	In the $\tSD$ variant of the $\SD$ problem, Alice and Bob are given $\{0,1\}$-matrices $\vec x, \vec y \in \{0,1\}^N \times \{0,1\}^N $, respectively.
	Under the promise that either there exists $t$ pairs of indices such that $x_{i,j}=y_{i,j}=1$, or that there exists $0$ such indices.
	The goal of Alice and Bob is then to distinguish between these two cases.
	We will denote the set of intersections by $\vec z$, where $\vec z_{i,j} = \vec x_{i,j} \wedge \vec y_{i,j}$.
\end{theorem}

The idea  is to construct an embedding of the $\tSD$ communication problem to a graph $G_{\vec z}$, such that the following holds. First, every query performed on $G_{\vec z}$ can be answered by exchanging $B$ bits of communication for a constant $B$. Second, one can solve the given $\tSD$ instance by sampling uniformly distributed copies of $H_D$ in $G_{\vec z}$. The parameter $t$ in the $\tSD$ problem is set according to $m,\cn h$ and the counts of the basic components of $D$, to ensure that  the lower bound on the communication complexity problem implies the desired lower bound specified in Theorem~\ref{thm:main_lb}.

\begin{theorem}[Corollary 2.7 in~\cite{Eden2018-approx}]
The communication complexity of $\tSD$ is $\Omega(N^2/t)$.
\end{theorem}
We shall prove that the problem of $\tSD$ can be reduced to the problem of estimating the number of copies of $H$ in a graph $G_{\vec z}$,
such that each query in $G_{\vec z}$ can be answered in constant time.
Namely, we prove that for a given $\HH$, the graph $G_{\vec z}$ consists of several gadgets, that are independent of the instance $(\vec x, \vec y)$,
and a \CC{} gadget that embeds the instance $(\vec x, \vec y)$ to the graph $G_{\vec z}$ as follows.
If $(\vec x,\vec y)$ intersect, then at least a constant factor of the copies of $H_D$ in $G_{\vec z}$ are contributed by this gadget, and otherwise this gadget contributes no copies.
The family of graphs $\mG$ is then defined to be the collection of graphs $\{G_{\vec z}\}$ for all possible $\vec z$ that are the intersection  of an \tSD\ instance. 
Thus by uniformly sampling copies of $H_D$,
one can distinguish between the case that $\vec x,\vec y$ are disjoint to the case where they intersect (by sampling a constant number of copies and checking if some are contributed by the \CC).
It follows that for every $N$ and $t$, $\Omega(N^2/t)$ queries are required in order to sample uniform copies of $H$. 

Our lower bound theorem is very generic as it works for any decomposition that contains at least one cycle, and for a variety of plausible basic component counts (those that meet the constraints specified in Definition~\ref{def:good_cnts}). Hence, we shall start with a (sketched) proof for a specific easy basic case. 
%Therefore, before proving the theorem for the general case we consider a specific given decomposition: $D=\{O_3, S_p\}$. 
The ideas in proving the general case will be the same, however due to the generality of the statement, many technical difficulties arise in satisfying all counts simultaneously. Hence, we defer that analysis of the general case to Subsection~\ref{sec:lb_general}.

\subsection{Warm up: a lower bound for a decomposition $D=\{O_3, S_p\}$}\label{sec:warm_up}

In this section we prove the first term in our lower bound for a specific decomposition,  $D=\{O_3, S_p\}$ and for the case that lower bound is sublinear in $m$, and the max cost in the bound is due to the $O_3$ component. 
\begin{theorem}
	\label{thm:warm_up}
	Let $D=\{O_3, S_p\}$ be a decomposition and assume that we are given the counts $n,m,\cn o_3, \cn s_p$ and $\cn h$. Further assume that the counts are such that $\cn s_p\geq \sqrt m^{p+1}$, $\max\{cost(O_3), cost(S_p)\}=cost(O_3)$ and $\cn h\geq \sqrt m\cdot \cn s_p$.
	Then there exist a motif $H_{D}$ with decomposition $D$, and a family of graphs $\mG$ such that for every 
	 $G\in \mG$ the   counts  are as above (up to constant factors), and such that  sampling a uniformly distributed copy of $H_{D}$ in a uniformly chosen $G\in \mG$ requires 
	\[
	\Omega\left(\max_{i}(cost(C_i))\cdot \frac{\cn o_3\cdot \cn s_p}{\cn h}\right)=  \Omega\left(\frac{m^{3/2}\cdot \cn s_p}{\cn h}\right)
	\]
	queries in expectation.
\end{theorem}
\begin{proof}[Proof Sketch.]
	By the above it  holds that $\cn h\leq \prod_{i} \cn c_i$, we let $\alpha= \prod_{i}\cn c_i /\cn h$ so that $\alpha>1$.
	 We shall rearrange the lower bound:
 \[ 	\frac{m^{3/2}\cdot \cn s_p}{\cn h}=\frac{m^{3/2}}{\cn o_3} \cdot \frac{\cn o_3\cdot  \cn s_p}{\cn h}=\frac{m^{3/2}}{\cn o_3} \cdot \alpha =\frac{m^{3/2}}{\cn o_3/\alpha}.\]
	The family $\mG$ is the set of graphs  $\{G_{\vec z}\}$ for all possible vectors $\vec z=\vec x\cdot \vec y$ where   $(\vec x, \vec y)$ are instances of the \tSD\ problem, for a value $t$ that will be set shortly.
	Fix an instance $(\vec x, \vec y)$ of \tSD\, and let $\vec z=\vec x\cdot \vec y$.
%	To prove the lower bound we ``hide'' in $G$  a set $T$ of $|T|=o_3/\alpha$ triangles that contribute $\Theta(\cn h)$ copies of $H_{D}$ to $G$. Then, any algorithm that outputs a uniformly distributed copy of $H_{D}$ can be used to detect the existence of these triangles. 
%	Previous lower bounds by~\cite{Eden2018-approx, AKK19} state that to output a uniformly distributed copy of a triangle in a graph $G$ with $m$ edges and $|T|$ triangles requires $\Omega(m^{3/2}/|T|)$ queries in expectation. Unfortunately, we cannot directly use these bounds as the graph $G$ must contain additional $\cn o_3\geq |T|$ triangles to meet the constraints of the theorem, and therefore this is not sufficient for our needs. However, we can directly use a reduction from $\tSD$ to  prove that detecting this specific set of triangles also takes $\Omega(m^{3/2}/|T|)$ queries in expectation. 
	We shall describe an embedding from $\vec z$ to $G_{\vec z}$ so that sampling a uniformly distributed copy of $H_D$ in $G_{\vec z}$ solves $\tSD$ on $(\vec x, \vec y)$.
	We set $t=\left\lfloor |T|/\sqrt m \right\rfloor=\left\lfloor\cn o_3/(\sqrt m\cdot \alpha)\right\rfloor$ so that $\frac{m^{3/2}}{o_3/\alpha}=\frac{m}{t}$ and we consider the case that $N=\sqrt m$, so that $\Omega(N^2/t)=\Omega(m/t)=\Omega(m^{3/2}/(o_3/\alpha))$.
	Observe that this setting is valid since, by the assumption that the complexity is sublinear in $m$, it holds that $\frac{m^{3/2}\cdot \cn s_p}{\cn h}\leq m$, implying that $h\geq \sqrt m \cdot \cn s_p $. Therefore, $\frac{\cn o_3\cdot \cn s_p}{\alpha}\geq \sqrt m \cdot \cn s_p$, and it follows that $ o_3/(\sqrt m\cdot \alpha)\geq 1$ so that $t\geq 1$.
	
	We let $H_{D}$ be the motif of a triangle connected by a single edge to a star $S_p$.
	To describe the graph $G_{\vec z}$, we describe a corresponding gadget to each of the components $O_3$ and $S_p$ in $D$.
	The gadget corresponding to the star is a bipartite graph over two sets $R_1, R_2$ such that $|R_1|=1$ and $|R_2|=\cn s_p^{1/p}$ (if $\cn s_p > n^p$, then we can modify $R_1$ to be of size $\lfloor\cn s_p/n^\rfloor$ and $R_2$ to be of size $n$). There is a complete bipartite graph between $R_1$ and $R_2$. 
	
	The gadget used to create the $|T|$ odd cycles of length $3$ has $3+2=5$ sets $R_1,R_2, R_3, R_1', R_2'$, each of size $\sqrt m$. There is a complete bipartite graph between the sets $R_1$ and $R_3$ and $R_2$ and $R_3$.  The edges between the sets $R_1, R_2, R_1', R_2'$ are determined according to the $\tSD$ instance $\vec x, \vec y$ as follows. 
	For every pair of indices $i,j\in \sqrt m$,  if $\vec x_{ij}=\vec y_{ij}=1$ then  we add the  edge  $(r_1^i,r_2^{j})$ and let 		as the $(j-i)\th$ edge of $r_1^i$ and $r_2^{j}$, 
	and the edge  $(r_1^{j},r_2^i)$ as the $(j-i)\th$ edge of $r_1^{j}$ and $r_2^i$.  
	We also add the edges $((r')_1^i, (r')_2^j)$ and $((r')_1^j, (r')_2^i)$ and label them as the $(j-i)\th$ edge of their endpoints.
	Otherwise,  we add the edges  $(r_1^i,(r'_1)^{j})$, $(r_1^j,(r_1')^i)$, $(r_2^i,(r'_2)^{j})$ and $(r_2^j,(r_2')^i)$ to the gadget, and label them as the $(j-i)\th$ edge of their endpoints.
	Hence, if $(\vec x, \vec y)$ is a \YES instance we get that the \CC\ has $t\cdot  {\sqrt m}^{k-2}$ odd cycles, and if it is a \NO instance then the gadget induces no cycles. See Figure~2(b) for an illustration.
	Furthermore, in both cases,  the degrees of all vertices in the gadget are exactly $2\cdot \cn o_k ^{1/k}$, and the ``gadget edges'' of the vertices in $R_1, R_2$ are their first $\sqrt m$ edges (in terms of edge labels).
	We furthermore add a complete bipartite graph between the two $R_1$ sets of the two gadgets.
	Observe that at this point, the count $\cn o_3$ is not satisfied as $G$ only contains $|T|< \cn o_3$ triangles. 
	As the set of counts is valid, there exists a graph $G'$ for which they are all satisfied.
    To finalize the construction, we add the graph  $G'$ to $G$ as a subgraph as a disconnected component. %See Figure~\ref{fig:warm_up} for an illustration of $H_D$ and $G$.
	
	\begin{figure}[t]\label{fig:warm_up}
		\centering	
		%\begin{minipage}{.36\textwidth}
		\begin{subfigure}{.33\textwidth}\label{fig:warm_up_H}
			\centering
			\includegraphics[width=.8\textwidth]{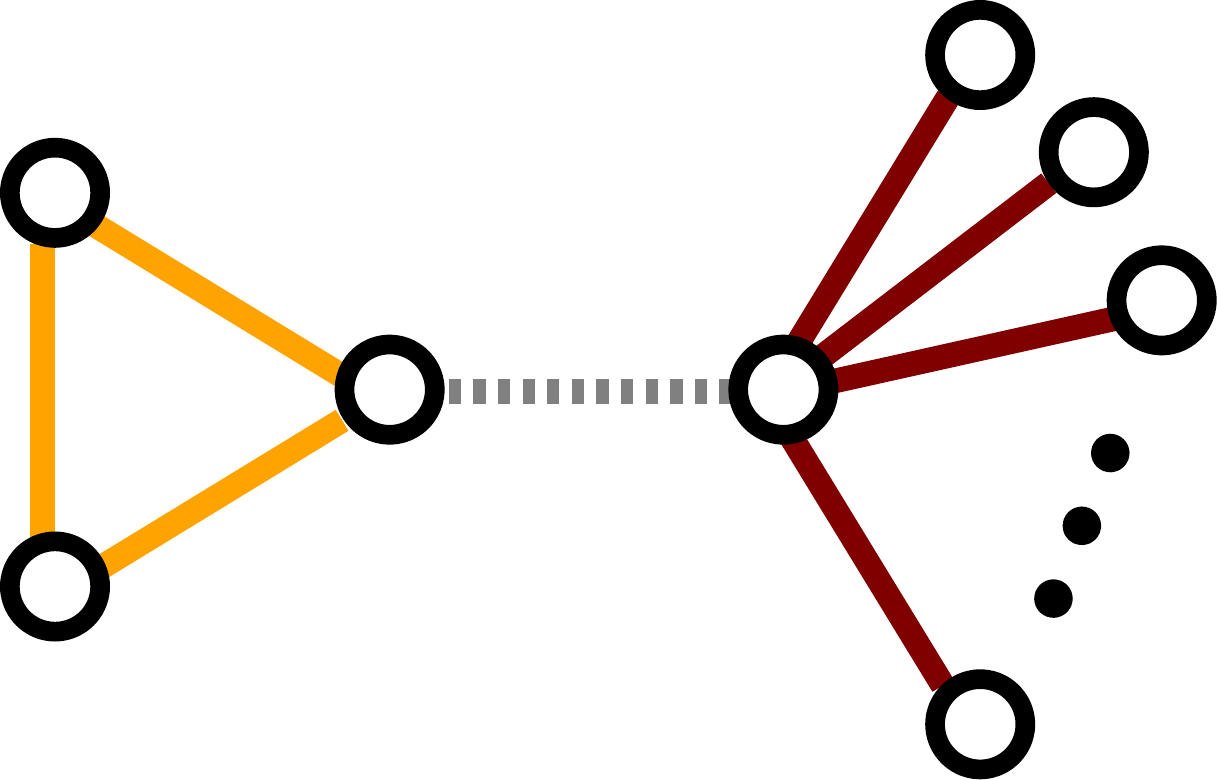}
%			\caption{$H_{D}$ for $D = \{O_3, S_p\}$.}
			\label{fig:example_D}
		\end{subfigure}%
		%\end{minipage}\qquad%
		\begin{subfigure}{\textwidth}
			\begin{minipage}{.66\textwidth}\label{fig:warm_up_G}
				\centering
				\includegraphics[width=\textwidth]{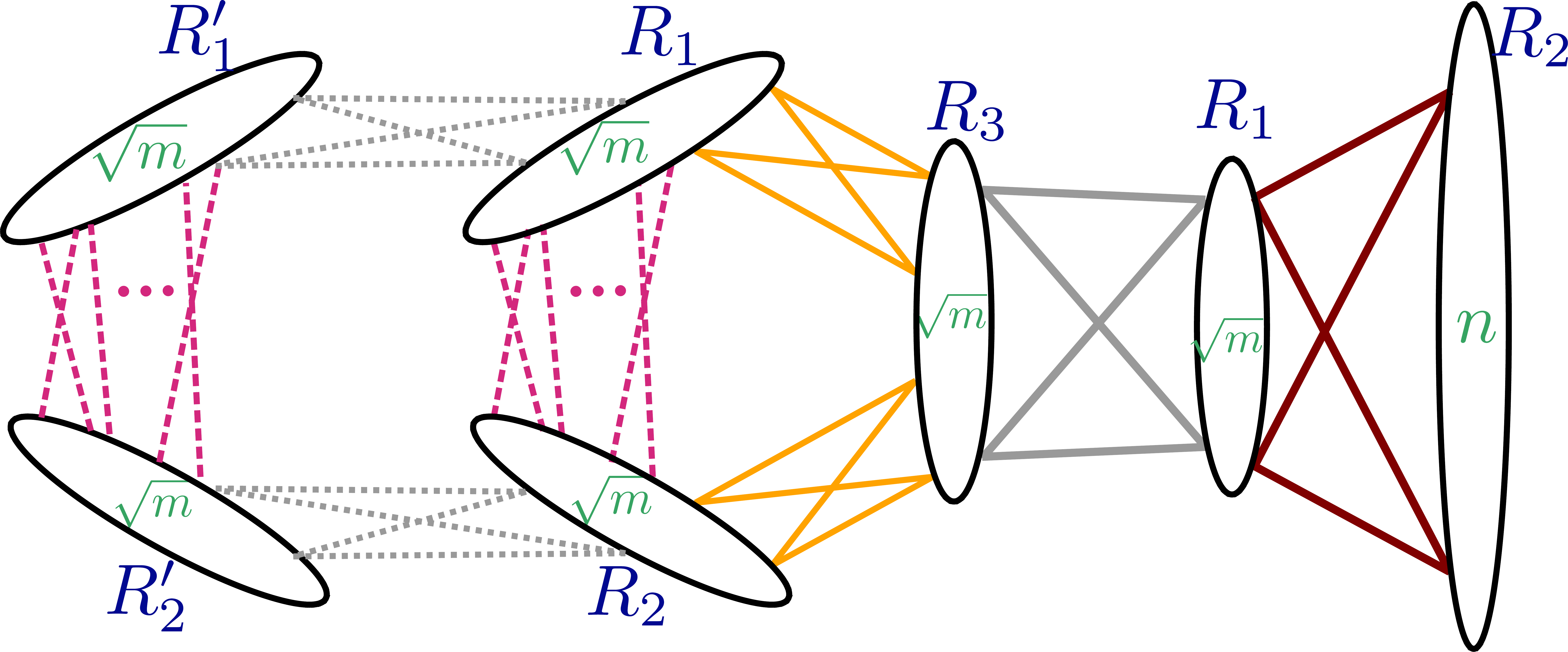}
%				\caption{The graph $G$ for $H$ as in Figure~\ref{fig:warm_up_H}.}
				\label{fig:example_H_D_lb}
			\end{minipage}
			%		\begin{minipage}{.33\textwidth}
			%			\caption{Complete lower bound construction for the decomposition shown in \cref{fig:example_D}, using all four types of gadgets
			%				(\cref{sec:lower_bound_gadgets}): clockwise from top, we have \CC, \SSS, \PP, \MM\\. \Ttodo{we reference this picture before the gadget names are defined}} 
			%			\label{fig:example_H_D_lb}
			%		\end{minipage}
		\end{subfigure}%
		\caption{
			(a) The motif $H_{D}$ for $D = \{O_3, S_p\}$ (b) The graph $G\setminus G'$.
			Orange/red crossed lines indicate a complete bipartite graph of intra-gadget edges,
			gray crossed lines indicate a complete bipartite graph of inter-gadget edges,
			and pink dotted lines indicate ``potential" edges -- i.e., ones whose existence depends on the $\tSD$ instance $\vec x,\vec y$.}
	\end{figure}

	By the construction of the gadgets, there are $\Theta(\cn s_p + \sqrt m^{p+1}) = \Theta(\cn s_p)$ copies of $S_p$ in the graph, as well as $\cn o_3$ triangles,  $\Theta(n)$ vertices and $\Theta(m)$ edges. Hence, the basic counts are satisfied (up to constant factors).

	By construction of the $O_3$ gadget, we have  that if $\vec x\cdot \vec y=0$, then the graph $G\setminus G'$ is bipartite, and otherwise it contains $t\sqrt m \cdot \cn s_p = (\cn o_3/\alpha)\cdot \cn s_p=\cn h$ many copies of $H_{D}$.
	Hence, given an algorithm $\mA$ that samples uniformly distributed copies of $H_{D}$,
	to solve the given $\tSD$ instance Alice and Bob proceed as follows. First they implicitly construct the graph $G_{\vec z}$ as described. Then, Alice and Bob both invoke $\mA$ using their shared randomness as the randomness of $\mA$ (so that  $\mA$ is now deterministic and Alice and Bob see the same queries during $\mA$'s run). Whenever $\mA$ queries $G_{\vec z}$, they either answer the query themselves (in case it does not depend on the input instance) or communicate $O(B)$ bits  to answer it.  They repeat this process for $10$ times. 
	Once all invocations of $\mA$ conclude, if all the returned copies of $H_{D}$ are from $G'$ then Alice and Bob respond that the input matrices are disjoint, and otherwise, they respond that the matrices intersect.  
	In case the matrices intersect, $1/2$ of the copies of $H_{D}$ are in $G\setminus G'$, and therefore, Alice and Bob respond incorrectly with probability $1/2^{10}$. If however the sets do not intersect, Alice and Bob respond correctly with probability $1$.
	
	Assume that each query can be answered by Alice and Bob  exchanging $O(B)$ bits of communication. Then the number of expected number of queries $Q$ performed by $\mA$ is lower bounded by $Q\cdot B=\Omega(m/t\cdot B)$, and for $B=O(1)$ we get $Q=\Omega\left(\frac{m^{3/2}\cdot \cn s_p}{B \cdot \cn h}\right)$.
	
	It remains to bound $B$. Here we only sketch the proof, as the full proof  is identical for this case and the general one, and it is given in Lemma~\ref{lem:comm_bits}. First observe that the degrees of all  vertices are determined independently of the input instance to $\tSD$. Indeed all vertices in  the cycle gadget have degrees $2\sqrt m$ and the structure of the star gadget does not depend on $\vec x, \vec y$.
	For a pair query $(u,v)$, unless both vertices belong to $R_1\cup R_2\cup R_1'\cup R_2'$ the answer is independent to the input instance. Otherwise, assume for example that $u=r^1_i\in R_1$ and $v=r^2_j\in R_2$. Then to answer the query, Alice and Bob send each other the bits $x_{ij}, y_{ij}$, and if they intersect they answer that the pair is an edge, and otherwise it is not. Other pair queries within these sets can be answered similarly, and so does neighbor queries on vertices in these sets.
	Hence, each query can be answered by exchanging  $O(B)=O(1)$ bits of communication, and we get 
	$Q=\Omega\left(\frac{m^{3/2}\cdot \cn s_p}{ \cn h}\right)$, as required.
	\iffalse	

	Now consider the case that the maximum component cost is achieved by $S_p$, so that we need to prove an $\Omega\left(\frac{\cn o_3\cdot m\cdot \min\{n^{p-1}, S_p^{(p-1)/p}\}}{\cn h}\right)$ lower bound. As before we let  $\alpha= \prod_{i}\cn c_i /\cn h$ so that $\alpha>1$. We focus on the case that $\min{n,s_p^{1/p}}=s_p^{1/p}$.
	Here too let us rearrange the lower bound term
	\[
	\Omega\left(\frac{\cn o_3\cdot m\cdot \cn s_p^{(p-1)/p} }{\cn h}\right)= \Omega\left(\frac{m}{\cn s_p^{1/p}}\cdot \frac{\cn o_3\cdot \cn s_p}{\cn h}\right)= \Omega\left(\frac{m}{(\cn s_p^{1/p}/\alpha)}\right).
	\]
	By the assumption that the bound is sublinear, we get that $\cn s_p^{1/p}\geq \alpha$. Also, since the cost is achieved by the $S_p$ component, we have that $m/\cn s_p^{1/p}\geq m^{3/2}/\cn o_3$. Let $\beta =\frac{m/\cn s_p^{1/p}}{m^{3/2}/\cn o_3}$. It holds that $\cn o_3 =\sqrt m \cdot s_p^{1/p}\cdot \beta\geq \sqrt m\alpha\beta$. For technical reasons, it would be easier to again prove the lower bound by embedding the $\tSD$ instance into the $O_3$ gadget. Hence we let $|T|=\cn o_3/(\alpha\beta)$ (so that $|T|\geq \sqrt m$)
	\fi
\end{proof}

\iffalse
	The  case that the maximum component cost is due to the  $S_p$ component requires a different construction, and therefore we will first prove the general lower bound for the case that the max cost component is an odd cycle, and only then return to this case.
\fi

\subsection{Proof of Theorem~\ref{thm:main_lb}}\label{sec:lb_general}
%We first present the lower bound for the case that the maximum cost $cost_i\{C_i\}$ is due to an odd cycle gadget. 
%and that the length $k_1$ of the odd cycle that maximizes the cost $cost(O_{k_i})$ is of unique length. 
%In sections~\ref{sec:multiple-min-length} and~\ref{sec:stars-lb}, respectively, we extend this results to the case that the decomposition contains multiple odd cycles of length $k_1$, and to the case that the decomposition consists solely of stars.
%In Section~\ref{sec:stars-lb}, we prove the theorem for  the case that the maximum cost is due to a star component.

%\begin{theorem}
%\end{theorem}

Let $D=\decomp$. 
To prove the lower bound of Theorem~\ref{thm:main_lb}, we first construct a graph $H_D$ with optimal decomposition is  $D$. We then construct a family of graphs $\mG$ such that each $G\in\mG$ satisfies all the counts and constraints of the theorem, and so that sampling a uniformly distributed copy of $H_{D}$ in a uniformly chosen $G\in \mG$ requires 
$\Omega\left( \min\left\{cost(O_k)\cdot \frac{\prod_{i} \cn c_{i}}{\cn h},m\right\}\right)$ samples. 

\textbf{Constructing the motif $H_{D}$.}
Given a decomposition $D$ we construct the graph $H_{D}$ as follows.
Recall that $O_{k_1}$ denotes the odd cycle with maximum cost, and denote its vertices by $v^{k_1}_1, \ldots v^{k}_{k_1}$.
If there exists a star $S_p$ in $D$ with count $\cn s_p>|H|\cdot \sqrt m^{p+1}$, then we connect its star center to one of the vertices of $O_{k_1}$. 
If for at least one of the cycles in $D$, $\cn o_{k}\leq  \sqrt m^{k-1}$, then we connect to it all the stars of $D$, except for the one that is connected to $O_{k_1}$.
We connect the rest of the components of $D$ with a single edge to $O_{k_1}$ , where stars are connected through their star center, and odd cycles are connected through arbitrary vertices in each of the cycles.

\textbf{Constructing the graph family of graphs $\mG$.} 
The basic structure of all graphs $G$ in the family $\mG$ will be the same, except for a small set of edges which will be determined according to the \tSD\ instance $(\vec x, \vec y)$, or more specifically, according to $\vec z=\vec x \cdot \vec y$.
To construct the family of graphs $\{G_{\vec z}\}$, we first define gadgets that correspond to the stars and odd cycles of $D$. 

We differentiate between \emph{short} odd cycles of length $k_i$ for $k_i\leq k_1$  (if such exist in $D$), and those with higher lengths than $k_1$. The reason is that we want the gadgets corresponding to short odd cycles to create $\cn O_k$ odd cycles, while not creating ``too many'' $k_1$ odd cycles. (This is also the reason behind constraint~\ref{const:short_cycles}.) 

\begin{itemize}
	\item \MM: Given $O_k$ and $\cn O_{k}$ such that $\cn o_k>\sqrt m^{k-1}$, this gadget is  a complete $k$-partite graph,
	comprising of sets of vertices $R_1, R_2,\cdots, R_{k}$, each of size $\Theta(\cn O_{k}^{1/k})$.
	Each adjacent pair $R_i, R_{i+1 (\mod k) }$ induces a complete bipartite graph.	
	(Observe that for every graph $\cn o_{k}\leq m^{k/2}$ and therefore for every $i\in [k]$, $|R_i|\leq \sqrt m$.)	
	
	\item \SCG: 
	 Given $O_{k}$ and $\cn o_k$ such that $\cn o_k\leq \sqrt m^{k-1},$
    this gadget has a set $R_1$ consisting of a single vertex $v_1$ and $k_1-1$ sets $R_i$ for $i\in[2,k-1]$, each of size $\cn o_{k_1}^{1/(k_1-1)}$. The sets form a $k_1$-tripartite motif. 
	\item
	\texttt{\SSS}:
	%	Given $S_p$ and $\cn s_p$, this gadget is a complete bipartite graph $R_1 \leftrightarrow  R_2$, where $R_1$ and $R_2$ as follows.
	%	The set  $R_2$ is  of size $\min\left\{(\cn s_p/\sqrt m)^{1/p}, n\right\}$, and the set $R_1$ is of size $\cn s_p/|R_2|^p$.
	%	(Observe that for every graph $\cn s_p\leq \sqrt m\cdot n^p$ and therefore for every $i\in [\ell]$, $|R_1|\leq \sqrt m$.)
	Recall that we assume that the counts $\cn s_{p_i}$ are either such that 
	there exists a cycle $O_{k}$ with length $\cn o_{k}\leq  \sqrt m^{k-1}$, or
	that each count $\cn s_p$  can be satisfied by a set $A$ of $\sqrt m$ numbers, $a_1, \ldots, a_{\sqrt m}$. That is, $\cn s_p=\sum_{i\in A} (a_i)^p$.
	
	In the former case, the star gadget is a bipartite motif $R_1\cup R_2$, where $|R_1|=|R_2|=n$ and the degrees of the vertices in $R_1$ are such that $\sum_{v\in R_1}d(v)^p=\cn s_p$. 
		Due to constraint~\ref{const:valid}, such a setting of degrees exists.
	The edges going from $R_1$ to $R_2$ are spread evenly among the vertices of $R_2$, so that $\forall r^2_i\in R_2,; d(r_2)\leq \davg$.
		
		In the latter case, the star gadget is a bipartite motif $R_1\cup R_2$, where $|R_1|=\sqrt m$ and $\forall r^1_i\in R_1,\; d(r^1_i)=a_i$. The set $R_2$ is of size $n$, and the edges from $R_1$ are distributed evenly among the vertices of $R_2$. 
\end{itemize}

To embed the $\tSD$ instance to $G_{\vec z}$, we use the following \CC\  that corresponds to $O_{k_1}$  which is (one of) the  maximum cost odd cycle in $D$. Since this gadget is used to distinguish the two families of graphs, it  appears in two forms,
corresponding to the \YES and \NO instance of the  problem.
\begin{itemize}
	
	\item
	\CC:
	This gadget will correspond to the odd cycle of length $k_1$ in  $H_{D}$ (a maximum cost odd cycle).
	The gadget contains $k_1$ sets $R_1, \mydots, R_{k_1}$ and  two additional sets $R_1', R_2'$, all of size $\sqrt m$. 
	Between every pair of sets $R_i,R_{i+1 (\; mod \sqrt m)}$, except between the pair $R_1,R_2$,   there is a  complete bipartite set. 
	The edges between the sets $R_1, R_2, R'1_, R'_2$ are determined according to the instance $(\vec x, \vec y)$ as follows.
	
	For every pair of indices $i,j\in \sqrt m$,  if $\vec x_{ij}=\vec y_{ij}=1$ then  we add the  edge  $(r_1^i,r_2^{j})$ 		as the $(j-i)\th$ edge of $r_1^i$, 
	and the edge  $(r_1^{j},r_2^i)$ as the $(j-i)\th$ edge of $r_1^{j}$ and $r_2^i$.  
	We also add the edges $((r')_1^i, (r')_2^j)$ and $((r')_1^j, (r')_2^i)$ and label them as the $(j-i)\th$ edge of their endpoints.
	Otherwise,  $\vec x_{ij}=\vec y_{ij}=0$, and we add the edges  $(r_1^i,(r'_1)^{j})$, $(r_1^j,(r_1')^i)$, $(r_2^i,(r'_2)^{j})$ and $(r_2^j,(r_2')^i)$ to the gadget, and label them as the $(j-i)\th$ edge of their endpoints.
	Hence, if $(\vec x, \vec y)$ is a \YES instance we get that 
	there are $t$ edges between $R_1$ and $R_2$, and so 
	the \CC\ has $t\cdot \sqrt m^{k-2}$ many  $k_1$ cliques. Otherwise, there are no edges between $R_1$ and $R_2$, and so the gadget is bipartite and induces no odd cycles.
%	Note that  the degrees of all vertices are independent of the input to the $\tSD$. 
\end{itemize}

See Figure~3(b) for an illustration of the different gadgets corresponding to the basic components of $H_D$.

Fix an input instance $\vec x, \vec y$ and let $\vec z=\vec x\cdot \vec y$.
The graph $G_{\vec z}$ contains one \CC{} that corresponds to the  $O_{k_1}$ component.  For any other $O_{k}, k\leq k_1$, if $k\leq k_1$ or  $\cn o_k \leq \sqrt m^{k-1}$, the graph contains a corresponding \SCG, and otherwise, 
the graph contains a \MM.
For all stars $S_{p_j}$ we add a \SSS. 
To connect the different gadgets, for each edge between two odd cycles, or between an odd cycle ant a star  in $H_{D}$, we add a complete bipartite graph between the two sets  $R_1$ of the corresponding gadgets. The way that the components of $D$ are connected, and the construction of  the gadgets of $G_{\vec z}$, ensure that this can be performed without exceeding $\Theta(m)$ edges between any two sets in $G_{\vec z}$. (Since all sets of odd cycle gadgets are of size $\sqrt m$, and since  $R_1$ sets of star gadgets with $|R_1|=n$ are only connected to sets $R_1$ of odd cycles for which $|R_1|=1$.)
Finally, we add to $G_{\vec z}$   a graph $G'$ for which all of the given counts are satisfied (recall there exists such a graph as we assume that the counts are valid). See Figure~3 for an illustration of a graph $G_{\vec z}$ for some $|\vec z|=t$ and motif $H_D$.

\begin{figure}[t]\label{fig:complete}
	\centering	
	%\begin{minipage}{.36\textwidth}
	\begin{subfigure}{.4\textwidth}
		\centering
		\includegraphics[width=.9\textwidth]{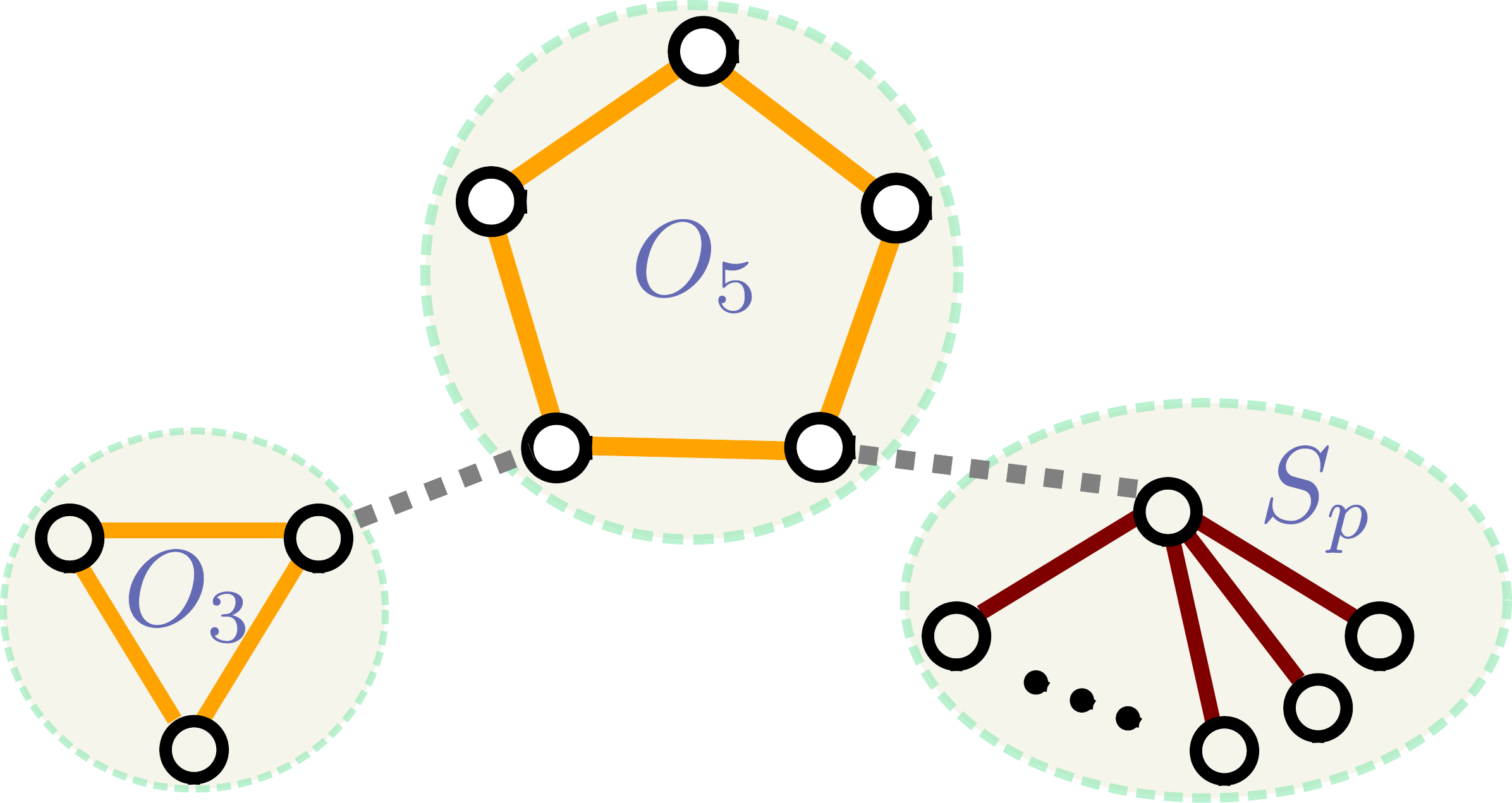}
		\vspace{1.72cm}
		\caption{motif $H_{D}$ for $D = \{O_3, O_5, S_p\}$}
		\label{fig:example_D}
	\end{subfigure}%
	%\end{minipage}\qquad%
	\begin{subfigure}{\textwidth}
		\begin{minipage}{.66\textwidth}
			\centering
			\includegraphics[width=\textwidth]{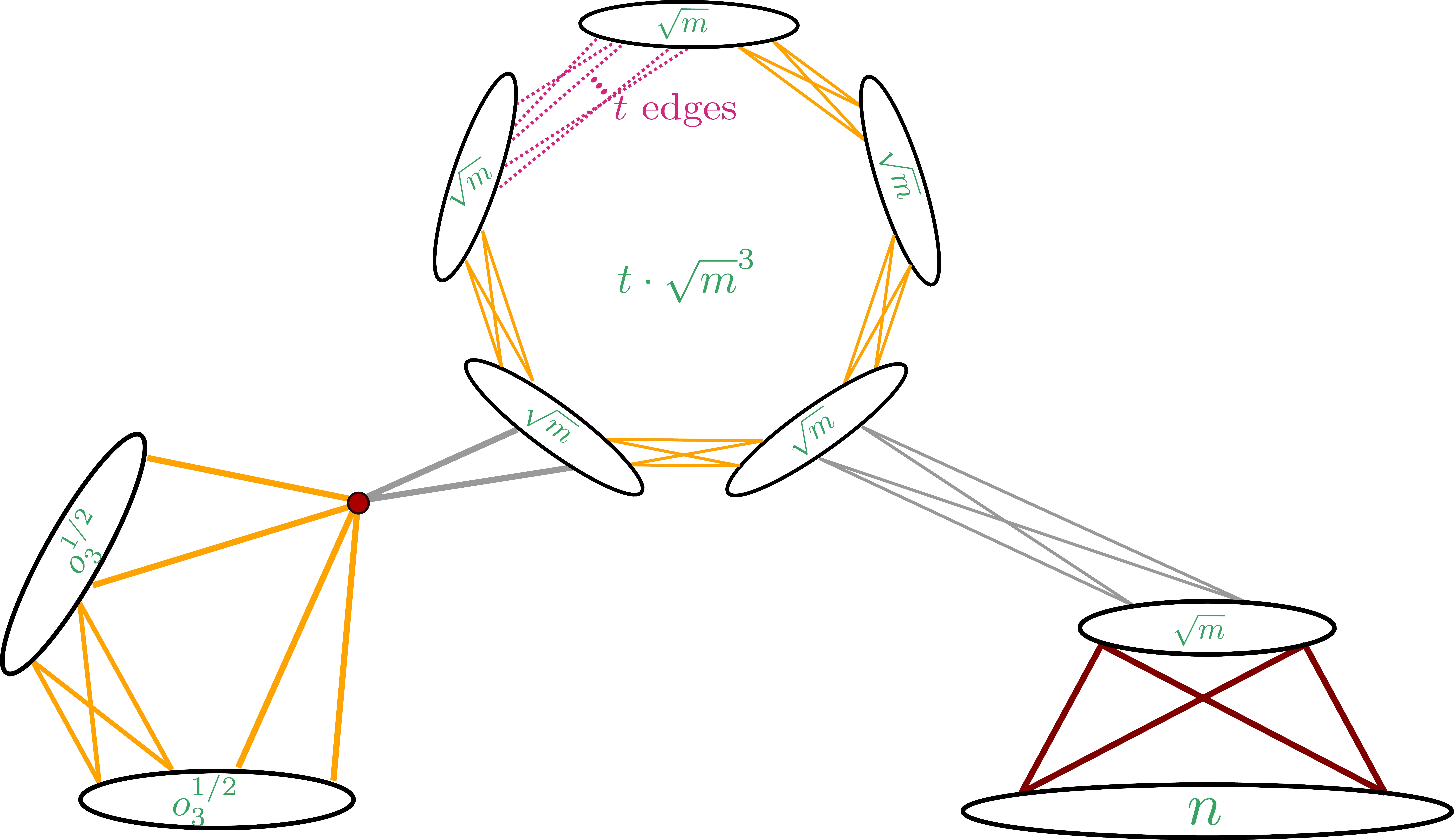}
			\caption{Complete lower bound construction of $G\setminus G'$ using all three types of gadgets:
				clockwise from top, we have \CC, \SSS, and \SCG.}
			\label{fig:example_H_D_lb}
		\end{minipage}
%		\begin{minipage}{.33\textwidth}
%			\caption{Complete lower bound construction for the decomposition shown in \cref{fig:example_D}, using all four types of gadgets
%				(\cref{sec:lower_bound_gadgets}): clockwise from top, we have \CC, \SSS, \PP, \MM\\. \Ttodo{we reference this picture before the gadget names are defined}} 
%			\label{fig:example_H_D_lb}
%		\end{minipage}
	\end{subfigure}%
	\caption{
		Orange/red crossed lines indicate a complete bipartite graph of intra-gadget edges,
		gray crossed lines indicate a complete bipartite graph of inter-gadget edges,
		and pink dotted lines indicate ``potential" edges -- i.e., ones whose existence depends on the $\tSD$ instance $\vec x,\vec y$.}
\end{figure}

\textbf{Proving the lower bound.}
We first consider the case that the $\comp\leq m$.
 As in the warm up case, we shall prove the lower bound by ``hiding'' $\cn h$ copies  of $H_D$ using a hidden set $T$  of $|T|$ $k_1$-odd cycles. That is, these $|T|$ odd cycles  will be added to the graph if and only if the matrices $\vec x$ and $\vec y$ intersect, and in turn they will create a constant number of copies of $H_D$ to $G_{\vec z}$.

We start by  rearranging  the lower bound terms and determining the values of $|T|$ and $t$. 
Let $\alpha =\prod_i c_i/\cn h$ so that $\alpha\geq 1$.
By the assumption that the lower bound is sublinear in $m$, we have that 
 $cost(O_{k_1})\cdot \alpha\leq m$, implying $\frac{m^{k_1/2}}{\cn o_{k_1}}\cdot \alpha\leq m\Leftrightarrow \cn o_{k_1}\geq \alpha m^{(k_1-2)/2}$.
Let \[|T|=\cn o_{k_1}/\alpha \text {\;\;\;\;\;\;and\;\;\;\;\;\;} t=\left\lfloor |T|/\sqrt m^{(k_1-2)/2} \right\rfloor= \left\lfloor \cn o_{k_1}/(\alpha{\sqrt m}^{(k_1-2)/2}) \right\rfloor\]
 so that $|T|\geq m^{(k-2)/2}$ and $t\geq 1$.
The lower bound we aim for is then 
\[
\max_{i\in[r]}\{cost(C_i)\}\cdot \frac{\prod_{i} \cn c_{i}}{\cn h} =
cost(O_{k_1})\cdot \alpha = \frac{m^{k_1/2}}{\cn o_{k_1}/\alpha}=\frac{m^{k_1/2}}{|T|}=\frac{m}{t}.
\]

%In the case that the maximum cost is due to the star motif $S_p$ (recall $p=\max_{j\in[\ell]}p_j)$, we let $\beta=cost(S_p)/cost(O_{k_1})$ (where $O_{k_1}$ is still the maximum cost odd cycle), and we aim to prove a lower bound of $cost(S_p)\cdot \alpha = cost(C_1)\cdot \alpha \cdot \beta =\frac{m^{k/2}}{\cn o_k/(\alpha\beta)}$. Hence, we set $|T|=\cn o_k/(\alpha\beta)$ and $t=\left\lfloor|T|/m^{(k-2)/2}\right\rfloor$. As before, one can verify that indeed $|T|>m^{(k-2)/2}$ and $t\geq 1$.

In order to prove the lower bound, we first prove that all the given motif counts of the basic components  are indeed as specified. 
That is, we prove the following lemma.
\begin{lemma}
	\label{lem:unique_min_lb_consistency}
	Let $G_{\vec z}$ be  as above. Then for any $\vec z$, 
	$G_{\vec z}$ contains $\Theta(n)$ vertices, $\Theta(m)$ edges,
	and $\Theta(\cn C_i)$ copies of each component $C_i$ in $D$.
\end{lemma}
\begin{proof}
	The graph $G'$ ensure that all counts are at least as specified. It remains to prove that the counts are not exceeded.
	
	Fix an odd cycle $O_{k_i}$. We shall verify that is count as is required.
	\begin{enumerate}
	    \item By construction, the gadget corresponding to $O_{k_i}$ contributes $\Theta(\cn o_{k_i})$ copies of $O_{k_i}$.
	    \item Now consider contributions from gadgets $O_{k_j}$ for $j\neq i$.
	    \begin{itemize}
	    \item If  $k_j>k_i$ then such gadgets do not contribute to $\cn o_{k_i}$, as a $k_j$-partite graph induces no $k_i$ odd cycles for $k_i<k_j.$
	    \item If $k_j\leq k_i$ and $\cn o_{k_j}\leq \sqrt m^{k_j-1}$, then by the construction of the \SCG\ and by constraint~\ref{const:cycles}, it contributes $(\cn o_{k_j})^{(k_i-1)/(k_j-1)}\leq \cn o_{k_i}$ odd cycles of length $k_i$. 
	    \item Finally, if $k_j\leq k_i$ and $\cn o_{k_j}> \sqrt m^{k_j-1}$, then by the construction of the \MM\ and by constraint~\ref{const:cycles}, it contributes $(\cn o_{k_j})^{(k_i)/(k_j)}\leq \cn o_{k_i}$ odd cycles of length $k_i$. 
	    \end{itemize}
	    \item If $z=\vec 0$, the \CC{} does not contribute any odd cycles, as it is bipartite.
	    Otherwise, when $z\neq \vec 0$, and $k_1>k_i$, the \CC{} also contributes $0$ odd cycles of length $k_i<k_1$.
	    If $k_i\geq k_1$, 
	    then the gadget contributes 
	     $\Theta(t\cdot {\sqrt m}^{k_j-2})$ odd cycles of length $k_i$.
    	Since the $O_{k_1}$ component is the odd cycle component with maximum cost, we have that 
	\[
	m^{k_1/2}/\cn o_{k_1} \geq  m^{k_i/2}/\cn o_{k_i} \;\;\;\Leftrightarrow\;\;\; 
	\cn o_{k_i}\geq \frac{\cn o_{k_1}}{m^{k_1/2}} \cdot m^{k_i/2} \;\;\;\Leftrightarrow\;\;\; \cn o_{k_i} \geq t\cdot m^{k_i/2+1}
	\]
	where the last inequality is  by the setting of $t=\cn o_{k_1}/{\sqrt m}^{k_1-2}$. 
	\end{enumerate}
	Hence, summing over all contributions from all the components, we get that the number of copies of $O_{k_i}$ is $\Theta(\cn o_{k_i})$.
	
	Now fix a star component $S_p$. The vertices of the odd cycle gadgets  contributes at most $\sqrt m^{p+1}\leq \cn s_p$ to the number of copies of $S_p$ in $G_{\vec z}$.  All star gadgets  contribute $\Theta(\cn s_p)$ contribute $\Theta(\cn s_p)$ copies of $S_p$.
	Hence, the number of $S_p$ stars in $G_{\vec z}$ is $\Theta(\cn s_p)$. 
\end{proof}

%We continue to prove that if the \tSD\ instance is a \YES instance, then there are $\Theta(\cn h)$ copies of $H_{D}$ in $G_{\vec z}$, and otherwise, if the instance is a \NO instance, then there no copies. Hence, solving \tSD can be reduced to estimating $\oh$ in $G.$

%\begin{lemma}
%If $\vec{z}=\vec 0$, then $G_{\vec z}$ has no copies of $H_{D}$. Otherwise, if $\vec x$ intersects $\vec y$ on $t$ indices, then for $t=\Ttodo{complete}$ $\mG_{mD}$ has $\Theta(\cn h)$ copies of $H_{D}$.
%\end{lemma}

\begin{lemma}\label{lem:h_counts}
    Let $\mG$ be the family of all graphs $G_{\vec z}$ such that $\vec z=\vec x\cdot \vec y$ for $\vec x, \vec y$ that are instances of \tSD.
    Let $B$ be an upper bound on the number of bits it takes Alice and Bob to communicate in order to answer queries on  any graph $G_{\vec z}\in \mG$.
	Then for any $\cn h$, and any algorithm that with high success  probability samples a uniformly distributed copy of  $H_{D}$ from a uniformly chosen $G_{\vec z}\in \mG$, the number of required queries is 
	\[\Omega\left(\frac{m^{k_1/2}\cdot \prod_{i>1}\cn c_i}{B\cdot \cn h}, m/B\right) \] 
	 in expectation,
	where $\cn h$ denotes the number of copies of $H_{D}$ in $G_{\vec z}$.
\end{lemma}
\begin{proof}%[Proof of Theorem~\ref{thm:single_cycle_lower_bound}]
	First assume that the first term achieves the minimum. In that case we have that $h\geq m^{k_1/2-1}\cdot \prod_{i>2} \cn c_i$ and we aim to prove a lower bound of $\Omega\left(\frac{1}{B}\cdot \comp \right)$.
	We let $t=  \left\lfloor \left(m^{(k_1-2)/2}\cdot \prod_{i>1} \cn c_i\right)/\cn h\right\rfloor$. This $t$ is the one  which determines the t-\SD{} communication problem we consider.
	Given a t-\SD{} instance with inputs $\vec x$ and $\vec y$, we construct $G_{\vec z}$ as described above, where recall that $\vec z=\vec x\cdot \vec y$ determines the \CC. 
	
	We first consider the case that  $|\vec z|=t$, and argue that the number of copies of $H_D$ in $G\setminus G'$ is $\Omega(\cn h)$.
	Since $|\vec z|=t$,  the \CC\ corresponding to $O_{k_1}$ contains $t\cdot \sqrt m^{(k_1-2)/2}$ odd cycles of length $k_1$ (since fixing an edge $t$, one can complete it to a $k_1$ length cycle by choosing one vertex (out of the possible $\sqrt m$) in each of the sets $R_i$ for $ i \in [3,k_1]$).
	By choosing one odd cycle or star from every odd cycle and star gadgets in $G$, it holds 
	that 
% 	Recall the set $R_1$ of the \CC\ is connected by a complete bipartite  graph to the sets $R_1$ of all other gadgets, and each of  gadget corresponding to a component $C_i$ contains $\cn c_i$ many $C_i$ components.
% 	Also, by the construction of $H_{D}$ any 
% 	copy of it must be such that   odd cycles of length $k_1$ are only  contributed by the CC gadgets (as was verified in the case that $\vec z=0$). 
	 the  number of copies of $H_{D}$ in $G_{\vec z}\setminus G$ is at least  $ t \cdot   m^{(k_1-2)/2} \cdot \prod_{i>1} \cn c_i = \left\lfloor |T|\cdot \prod_{i>1} \cn c_i \right\rfloor$. 
	 Observe that by the construction of $G$, the edges between the odd cycles and stars of different components agree with the non-decomposition edges of $H_D$.
	Hence, the  number of copies of $H_{D}$ in $G_{\vec z}\setminus G'$ is  at least $\Omega(\cn h)$.

	We now turn to the  case that  $\vec z=\vec 0$.
	and argue that the graph $G\setminus G'$ contains less  $o(\cn h)$ copies $H_D$.
	We deal separately the two potential cases due to constraint~\ref{const:short_or_star}, that is, that either there is at least one star with $\cn S_p>|H|\sqrt m^{p+1}$, or that that for all odd cycle components $O_{k_i}$ for  $k_i\leq k_1$, there are a few of them ($\cn o_{k_i}\leq \sqrt m^{k-1}$).
	(Recall that this constraint is to prevent short cycle gadgets from creating too many copies of $H_D$ within themselves.)
	
	Assume first that there exists at least one star $S_p$ in $D$ with $\cn S_p=\omega(m\cdot (\cn o_{k_*})^{(p+1/k_*)})$,
	where recall that $k_*$ is the index of the $O_k$ component that maximizes $\cn o_{k}^{1/k}$.
	Recall that by the construction of the motif $H_D$,  $S_p$ is connected to $O_{k_1}$. Also recall that in that case, the \SCG\ is identical to the \MM, and it holds that a \MM\ can potentially create at most $k_i \cdot (\cn o_{k_i})^{1/k_i}\cdot  (\cn o_{k_i})^{p/k_i}=k_i \cdot  (\cn o_{k_i})^{(p+1)/k_i}$ copies of $S_p$.
	Also, for all other $S_{p_j}\in D$, at most $ \sqrt m^{p+1}$ copies of 
	$S_{p_j}$ are created
	Hence, each \MM\ creates at most 
	\[
	 (\cn o_{k_i})^{(p+1)/k_i} \cdot \prod_{j\in[q]} (\cn o_{k_i})^{k_j/k_i}\cdot \prod_{j\in[\ell], S_{p_j}\neq S_p} 	 \sqrt m^{p_j+1} 	
	\leq 
	 (\cn o_{k_i})^{(p+1)/k_i}  \cdot \prod_{j\in[q]} \cn o_{k_j}\cdot \prod_{j\in[\ell], S_{p_j}\neq S_p} 	\cn  s_{p_j} \;,
	\]
 	where the last equality is due to constraint~\ref{const:stars_lb}.
	Also, since  $\cn o_{k_1}\in [\sqrt m^{k_1-2}, \sqrt m^{k_1}]$, and $t\geq 1$, it holds that $t\cdot \sqrt m^{k-2} \geq \cn o_{k_1}/m$.
	Hence, 
	\[h=t\cdot \sqrt m^{k-2}\cdot \prod_{i>1} \cn o_{k_i}\cdot \prod_{j\in[\ell]} \cn s_{p_j} \geq \frac{1}{m} \prod_{i\in[q]} \cn o_{k_i}\cdot \prod_{j\in[\ell]} \cn s_{p_j} =
	\frac{1}{m}\cdot \cn s_{p}\cdot  \prod_{i\in[q]} \cn o_{k_i}\cdot \prod_{j\in[\ell], S_{p_j}\neq S_p} \cn s_{p_j}\;.
	\]
	Since $\cn s_p =\omega(m\cdot (\cn o_{k_i})^{(p+1)/k_i})$, it holds that 
	the number of copies created by the \MM\ of $O_{k_i}$ is $o(\cn h)$.
	Therefore, in that case the number of copies of $H_D$ in $G\setminus G'$ is $o(\cn h)$. 
	
	In the case that there is no star $S_p$ with sufficiently many copies as above, we have that constraint~\ref{const:short_cycles} holds. In that case, 
	for every $O_{k_i}$, either (1) $k_i\leq k_1$, and so $O_{k_i}$ has a \SCG\ with a part $R_1$ consisting of a single vertex; or (2) $k_i>k_1$ and  $O_{k_i}$ has an \MM.
    In case (1), since the part $R_1$ of the \SCG\ has a single vertex no copies of $H_D$ can be created. In case (2), since $k_i>k_1$, no copies of odd length cycles of length $k_1$ are formed, and again no copies of $H_D$ can be created. 
    Also, no copies of $H_D$ can be created by combining odd cycles of different gadgets, since each \SCG\ can contribute at most one odd cycle, and \MM\ cannot  contribute short cycles, and so at least one odd cycle of length $k_i<k_1$ will be missing. 
    
    Therefore, in both cases of constraint~\ref{const:short_or_star}, the number of copies of $H_D$ in $G\setminus G'$ is $o(\cn h)$, as claimed.

	%	To compare the two $H_{D}$ counts, recall that by the constraints of the Theorem, $\forall i\in[2,q],\ \cn O_{k_1}^{k_i} \leq \cn O_{k_i}$ and that for all $j\in[\ell]$, $\sqrt m^{p_{j+1}}\leq  \cn s_{p_j}$. Also, $\cn O_{k^*}^{k_1/k^*}\leq m^{(k_1-2)/2} \Leftrightarrow \cn o_{k^*}\leq m^{k^*/2-k^*/k_1}$ which also holds by the constraints of the theorem. 
	%	Hence, 
	%	\[
	%	h_1=2k^*\prod_{i\in [q]}\cn o_{k^*}^{k_i/k*}\cdot \prod_{j\in [\ell]}  \cn s_{p_j}
	%	=2k^* \cdot \cn o_{k_*}^{k_1/k^*}\cdot \prod_{i\in[2,q]}\cn o_{k_1}^{k_i/k^*} \cdot \prod_{j\in[\ell]} \cn s_{p_j} 
	%	=o( C m^{(k_1-2)/2}\cdot \prod_{i\in[2,q]} \cn o_{k_i} \prod_{j \in [\ell]} \cn s_{p_j}),
	%	\]
	%	implying that $h_1\leq C\cdot \cn h$.
	
	Now let $\mA$ be any algorithm that samples returns a uniformly distributed copy of $H_{D}$. Then Alice and Bob can invoke  $\mA$ on the  (implicit) graph $G_{\vec z}$ and whenever  $\mA$ performs a query, by the assumption of the lemma, Alice and Bob can communicate $B$ bits to answer it. 
	Alice and Bob repeat the above for $10$ times.
	Let $Q$ denote the number of queries each invocation of $\mA$ performs. After $\mA$ concludes all its runs, if $\mA$ returns any copy of $H_{D}$ from $G_{\vec z}\setminus G'$, then Alice declares that $x$ and $y$  intersect, and otherwise she declares they do not. 
	Since the number of copies  of $H_{D}$ from $G_{\vec z}\setminus G'$ is at least $1/2$ of the number of copies in $G_{\vec z}$,
	each invocation of $\mA$ should return a copy of $H_{D}$ from $G_{\vec z}\setminus G'$ with probability at least $2/3$. 
	 Hence, the probability that $\vec z\neq \vec 0$ and no copy from $G_{\vec z}\setminus G'$  is returned is at most $(2/3)^{10}.$
	Therefore,  Alice and Bob can with high probability solve the $\tSD$ instance using $O(Q\cdot B)$ bits of communication. 
	By the $\Omega(m/t)$ expected communication lower bound for $\tSD$, it follows that $Q=\Omega(m/(t\cdot B))$. Since $t=\Theta(\cn h/( m^{k_1/2-1}\cdot \prod_{i>1} \cn c_i))$,
	we get an 
	\[
	\Omega\left( \frac{m}{t\cdot B}\right)= \Omega\left(\frac{m^{k_1}/2\cdot \prod_{i>1}\cn c_i}{B \cdot \cn h} \right) 
	\]
	lower bound, as claimed.

	For the case that the minimum in the lower bound is due to the term $m$, we use the same proof, but with adjusted values of $|T|,t$ and the sizes of the sets in the \CC{} of $O_{k_1}$. All other arguments remain the same. 
	Recall that $\cn h = \prod_{i} \cn c_i/\alpha$, and so in this case we have that $\frac{m^{k_1/2}}{\cn o_k}\cdot \alpha \geq m \Rightarrow \cn o_k \leq \alpha \cdot m^{(k_1-2)/2}$.
	Let $\beta>1$ be  $\beta=\alpha m^{(k_1-2)/2}/\cn o_{k_1} \Rightarrow 
	 \cn o_k = \alpha \cdot (m/\beta)^{(k_1-2)/2}$.  
	We change the \CC{} that corresponds to $O_{k_1}$ by changing the sizes of its sets $R_3,\mydots, R_{k_1}$ to be of size $\sqrt m/\beta$ instead of $\sqrt m$. We now let
	 \[|T|=\cn o_{k_1}/\alpha=(m/\beta)^{(k_1-2)/2} \text {\;\;\;\;\;\;and\;\;\;\;\;\;} t=\left\lfloor |T|/\sqrt m^{(k_1-2)/2} \right\rfloor=1.\]
	 	By the same arguments as for the previous case, we have that if $\vec z=0$, then all copies of $H_{D}$ are in $G'$, and otherwise, $G_{\vec z} \setminus G'$ has $t\cdot (\sqrt m/\beta)^{k_1-2}\cdot \prod_{i>1} \cn c_i = (\cn o_{k_1}/\alpha)\cdot \prod_{i>1}\cn c_i=\prod_{i}\cn c_i /\alpha = \cn h$ many copies of $H_{D}$.
	 	Therefore, the proof continues as before and we get a lower bound of $\Omega(m/B\cdot t)=\Omega(m/B)$ on the expected query complexity of any algorithm that returns a uniformly distributed copy of $H_{D}$.
\end{proof}

It remains  to prove that queries on $G_{\vec z}$ can be answered by Alice  efficiently.

\begin{lemma}\label{lem:comm_bits}
	Alice  can answer any query to $G_{\vec z}$ using $O(1)$ bits of communication between Alice and Bob. That is $B=O(1)$.
\end{lemma}
\begin{proof}
	We consider each of the possible queries.
	
	\textbf{Answering degree queries and uniform edge sample queries.}
	Observe that all the vertices' degrees in the graph  are set regardless of the $(x,y)$ instance. Therefore, Alice knows the degree sequence and can produce a uniform edge sample and answer a degree query with zero communication.
	
	\textbf{Pair queries.}
	Pair queries that include at most one vertex from the sets $R_1, R_2, R'_1, R'_2$,  of the \CC\  can be answered with zero communication. 
	Pair queries $(u,v)$ where say $u=r_1^i$ and $v=r_2^j$, are answered as follows.
	Bob sends to Alice the bit $y_{i,j}$. If the two bits intersect then the answer to the pair query is positive and otherwise, it is negative. Queries on other pairs with both endpoints in $R_1, R_2, R'_1, R'_2$ are answered similarly.
	
	\textbf{Answering $i\th$ neighbor queries.}
	First, any neighbor queries for vertices outside \CC{} can be answered with zero communication. Let $(v, j)$ be an $j\th$ neighbor query for some $v$ in the \CC. If $v\notin R_1\cup R_2$ or $j> \sqrt m$ then again the query can be answered with no communication. Therefore, assume without loss of generality that $v=r_1^i$ for some  $r_1^i\in R_1$ and that  $j\leq \sqrt m$.
	In this case Bob will send the bit $y_{j+i}$ to Alice (recall that both Alice and Bob invoke the same algorithm using their shared randomness, so that the queries are known to both without communication).
	If $x_{i,i+k}\cdot  y_{i,i+k}=1$, then Alice answers $r_2^{i+j \mod \sqrt m }$. Otherwise, Alice answers $(r')_1^j$.
	Neighbor queries on vertices in $R_2, R'_1$ and $R'_2$ are answered similarly.
\end{proof}

Theorem~\ref{thm:main_lb} follow  from Lemma~\ref{lem:h_counts} and Lemma~\ref{lem:comm_bits}.

	\subsection{From Theorem~\ref{thm:main_lb} to Theorem~\ref{thm:lb_dc}}\label{sec:dc_main}

\begin{lemma}\label{lem:main_to_dc}
	Theorem~\ref{thm:lb_dc} follows from Theorem~\ref{thm:main_lb}.
\end{lemma}
\begin{proof}
	
 Assume that Theorem~\ref{thm:main_lb} holds. 
 Fix  $D$ to be a decomposition that contains at least one odd cycle component and a unique minimum odd length cycle, and fix   $n,m$ and a realizable value of $\textproc{dc}$. We would like to argue that there exists a motif $H_D$ with optimal decomposition $D$, and  a hard family of graphs $\mG$ over $n$ vertices, $m$ edges and with decomposition cost $\textproc{dc}$, such that sampling a uniformly distributed copy of $H$ in graphs uniformly chosen in $\mG$ takes $\Omega(\min\{\textproc{dc},m\})$. 
 In order to do so we shall specify a set of good counts. 
% Let $O_{k_1}$ be the minimum length odd cycle component in $D$.
% First we set $\cn h=\frac{m^{k_1/2}}{}$
We set the counts depending on the value of $\textproc{dc}$.
If $\textproc{dc}\geq  m$, then 
we set the odd cycle counts as follows:
for every $O_{k_i}\in D$,
\[
\begin{cases}
\cn o_{k}=\lceil m^{k/2}/\textproc{dc}\rceil & \text{ if $k_i=k$}
o_{k_i}=\sqrt m^{k_i-1}, \text{ if $k_i<k$}\\
o_{k_i}=\sqrt m^{k_i}, \text{ if $k_i>k$    } \;.
\end{cases}
\]
If $\textproc{dc}< m$, then 
we set the odd cycle counts as follows.
Let $O_{k'}$ be the minimum length odd cycle in $D$. for every $O_{k_i}\in D$,
\[
\begin{cases}
\cn o_{k_i}=\lceil m^{k'/2}/\textproc{dc}\rceil  \text{ if $k_i=k'$ }\\
o_{k_i}=\sqrt m^{k_i-1}& \text{ if $k_i>k'$ }\;.
\end{cases}
\]
Observe that by the assumption that there is only one odd cycle component of minimum length, indeed for every $k_i$, either $k_i=k'$ or $k_i>k'$.
In both cases we also set $\cn s_p=\davg \cdot n^p$
We also set $\cn h=\prod_{i\in [r]} \cn c_i$.
 
 In order to be able to invoke Theorem~\ref{thm:main_lb}, we argue that these counts are good, as defined in Definition~\ref{def:good_cnts}.
 First, to see that the counts are realizable, consider a graph 
$G$ which has a \SCG\ for every  $O_{k_i} \in D$ such that $k_i\leq  k$, and a \MM\ for  every  $O_{k_i} \in D$ such that $k_i> k$. For every $S_p\in D$ we have a \SSS. We let $H_D$ be the components of $D$ that are connected is some tree like manner, and we connect the  gadgets of $G$ by a complete bipartite graph between any two gadgets whose corresponding components in $H_D$ are connected.
It holds that the number of copies of $H_D$ in $G$ is $\cn h=\pi{\cn c_i}$.
One can  verify that in both cases of possible values of $\textproc{dc}$, the rest of the constraints of Definition~\ref{def:good_cnts} also hold. 

Finally, in case that that $\textproc{dc}>m$, $\max_{i\in [r]}{cost(C_i)}= m^{k/2}/\cn o_k=\Theta(\textproc{dc})$, and otherwise $\max_{i\in [r]}{cost(C_i)}=m^{k'/2}/\cn o_{k'}=\Theta(\textproc{dc})$.
Hence, we get that in both cases, 
\[
\dc(G,H_D, D)=\max_{i\in [r]}{cost(C_i)}\cdot \frac{\prod_{i\in[r]}\cn c_i}{\cn h}= \Theta(\textproc{dc}).
\]
Therefore, 
 we can invoke Theorem~\ref{thm:lb_dc}, and the theorem follows.
\end{proof}

	%\fi
	
	%\input{810-stars-cost-lb}
	
	%\input{805-lower_bound_construction}
	%\input{810-lower_bound_constructing_H}
	%\input{825-lower_bound_CC_construction.tex}
	%\input{830-lower_bound_bounding_count_H}
	%\input{880-lower_bound_final}
	%\input{960-star-collections-lb}
	
	\section{Acknowledgments}
	Talya Eden is thankful to Dana Ron and Oded Goldreich for their valuable  suggestions regarding the presentation of the lower bound results. The authors are thankful for the anonymous reviewers for their useful comments and observations. 
	
	\bibliography{ref}
	\appendix
	\section{Related Work}\label{sec:related-work}
We note that some of the works were mentioned before, but we repeat them here for the sake of completeness.
Over the past decade, there has been a growing body of work investigating the questions of approximately counting and sampling motifs in sublinear time.
These questions were considered for various motifs $H$, classes of $G$, and query models.

The study of sublinear time estimation of motif counts was initiated by the works of Feige~\cite{feige2006sums} and of Goldreich and Ron~\cite{GR08} on approximating the average degree in general graphs. 
Feige \cite{feige2006sums} investigated the problem of estimating the average degree of a graph, denoted $\davg$, when given query access to the degrees of the vertices. By performing a careful variance analysis, Feige proved that $O\left(\sqrt {n/\davg}/\epsilon\right)$ queries are sufficient  in order to obtain a $(\frac{1}{2}-\epsilon)$-approximation of $\davg$.
He also proved that a better approximation ratio cannot be achieved in sublinear time using only degree queries. The same problem was then considered by Goldreich and Ron~\cite{GR08}. Goldreich and Ron proved that an $(1+\epsilon)$-approximation can
be achieved with $O\left(\sqrt{n/d_{avg}}\right) \cdot \poly(1/\epsilon,\log n)$ queries, if neighbor queries are also allowed.
Building on these ideas, Gonen et al.~\cite{GRS11} considered the problem of approximating the number of 
$s$-stars in a graph. Their algorithm only assumed neighbor and degree queries. 
In ~\cite{counting_stars_edge_sampling}, Aliakbarpour, Biswas, Gouleakis, Peebles, and Rubinfeld and Yodpinyanee considered the same problem of estimating the number of $s$-stars in the  augmented edqu queries model, which allowed them to circumvent the lower bounds of~\cite{GRS11} for this problem. In~\cite{ERS19-sidma}, Eden, Ron and Seshadhari again considered this problem, and presented improved bound for the case where the graph $G$ has bounded arboricity.
In~\cite{ELRS, ERS18, ERS20-soda}, Eden, Ron and Seshadhri considered the problems of estimating the number of $k$-cliques in general and in bounded arboricity graphs, in the general graph query model, and gave matching upper and lower bounds. 
In~\cite{TetekTriangles}, T\v{e}tek considers both the general and the augmented query models for approximately counting triangles in the super-linear regime.
 In~\cite{Eden2018-approx}, Eden and Rosenbaum presented a framework for proving motif counting lower bounds using reduction from communication complexity, which allowed them to reprove the lower bounds for all of the variants listed above.

In~\cite{ER18,ERR19}, Eden and Rosenbaum and Ron has initiated the study of sampling motifs (almost) uniformly at random. They considered the general graph query model, and  presented upper and matching lower bounds up to $\poly(\log n/1/\eps)$ factors,  for the task of sampling edges almost uniformly at random, both for general graphs and bounded arboricity graphs.
Recently, T\v{e}tek and Thorup~\cite{tetek} presented an improved analysis which reduced the dependency in $\eps$ to $\log(1/\eps)$. This result implies that for all practical applications, the edge sampler is essentially as good as a truly uniform sampler.
They also proved that given access to what they refer to as hash-based neighbor  queries, there exists an algorithm that samples from the exact uniform distribution.
The authors of \cite{ERR19} also raised the question of approximating vs. sampling complexity, and gave preliminary results that there exists motifs $H$ (triangles) and classes of graphs $G$ (bounded arboricity graphs)  in which approximating the number of $H$'s is strictly easier than sampling an almost uniformly distributed copy of $H$. This question was very recently resolved by them, proving a separation for the tasks of counting and uniformly sampling cliques in bounded arboricity graphs~\cite{eden2020almost}.

A significant result was achieved recently, when Assadi, Kapralov and Khanna gave an algorithm for approximately counting the number of copies of any given general $H$, in the edge queries augmented query model. They also gave a matching lower bound for the case that $H$ is an odd cycle. 
Fichtenberger, Gao and Peng presented a cleaner algorithm with a mich simplified analysis for the same problem, that also returns a uniformly distributed copy of $H$.

%Specifically, they provide an algorithm that estimates the number of occurrences of any arbitrary motif $H$ in $G$, denoted by $C_{H}$, to within a $(1\pm \epsilon)$-approximation with high probability. The running time of their algorithm is  $O^*\left(\frac{m^{\rho(H)}}{\#H}\right)$, where $\rho(H)$ is the fractional edge cover of $H$.\footnote{The fractional edge cover of a graph $H=(V_H, E_H)$ is a mapping $\psi:E_H \rightarrow [0,1]$ such that for each vertex $a \in V_H$, $\sum_{e \in E_H, a \in e} \psi(e) \geq 1.$ The fractional edge-cover number $\rho(H)$ of $H$ is the minimum value
%	of $\sum_{e \in E_H} \psi(e)$ among all fractional edge covers $\psi$.} The authors also prove that their bound is ``existentially optimal''. That is, they prove that for any edge cover number $k$, there exists a motif $H$ such that $\rho(H)=k$, for which their bound is optimal (while for other motifs, it is known that their bound is not optimal). 

Another query model  was suggested recently by Beame et al.~\cite{beame2017edge},
which assumes access to only \emph{independent set} (IS) queries or \emph{bipartite independent set} (BIS) queries . Inspired by group testing, IS queries allow to ask whether a given set $A$ is an independent set, and BIS queries allow to ask whether two sets $A$ and $B$ have at least one edge between them. 
In this model they considered the problem of estimating the average degree and gave an $O(n^{2/3})\cdot \poly(\log n)$
algorithm using IS queries, and $\poly(\log n)$ algorithm using BIS queries. Chen, Levi and Waingarten~\cite{Chen-IS-edges} later improved the first bound to $O(n/\sqrt m)\cdot \poly(\log n)$ and also proved it to be optimal.

\iffalse
Additional work on sublinear algorithms for estimating other graph parameters include those for approximating the size of the minimum weight spanning tree \cite{DBLP:journals/siamcomp/ChazelleRT05, DBLP:journals/siamcomp/CzumajS09, DBLP:journals/siamcomp/CzumajEFMNRS05}, maximum matching \cite{nguyen2008constant, yoshida2009improved}, and of the minimum vertex cover \cite{DBLP:journals/tcs/ParnasR07,nguyen2008constant,DBLP:journals/talg/MarkoR09, yoshida2009improved, hassidim2009local, onak2012near},
approximating the spectrum of a graph~\cite{cohen2018approximating}. 
\fi

	%\input{650-stars-ub}

\end{document}